\documentclass[11pt,a4paper]{article}
\usepackage[margin=2.6cm]{geometry}
\usepackage{macros_vaib}
\usepackage{tocloft}
\newtheorem{theorem}{Theorem}[section]

\newtheorem{proposition}[theorem]{Proposition}
\newtheorem{lemma}[theorem]{Lemma}
\newtheorem{assumption}[theorem]{Assumption}

\newtheorem{remark}[theorem]{Remark}

\newtheorem{definition}[theorem]{Definition}

\title{Energy-optimal predictive control of 
discrete-time port-Hamiltonian systems: Closed-loop practical stability}
\author{Vaibhav Kumar Singh, Manuel Schaller, Timm Faulwasser, Karl Worthmann
\thanks{The authors gratefully acknowledge funding by the Deutsche Forschungsgemeinschaft (DFG, German Research Foundation) – Project-ID 519323897.}}

\date{}
\begin{document}
\maketitle
\begin{abstract}
We study a dissipativity-based model predictive control (MPC) scheme for energy-optimal constrained output stabilization of a discrete-time nonlinear SISO port-Hamiltonian system described by difference and differential representation. For the optimal control problem to be solved in each MPC step, we establish a measure turnpike behavior. The turnpike set in our work is obtained from the underlying port-Hamiltonian structure without assuming {existence} of steady-states  or periodic orbits. In particular, unlike existing results which typically establish the turnpike property w.r.t.\ a controlled forward invariant set such as optimal steady-states or optimal periodic orbits, we do not require the turnpike set to be forward control invariant. For MPC, we establish recursive feasibility and show practical stability to the turnpike set w.r.t.\ the MPC closed loop leveraging both dissipativity and the port-Hamiltonian structure. Finally, we demonstrate our results using two numerical examples.
\end{abstract}

\section{Introduction}\label{se:intro} \vspace{0mm} 
\ \ \ Port-Hamiltonian (pH) systems provide a systematic and modular framework for modeling, analysis, and control of physical systems governed by energy exchange, interconnection, and dissipation \cite{SchaJelt14, ScMa:13}. Their defining structure separates power-conserving interconnection mechanisms from dissipative components and external ports through which energy is exchanged between the system and its environment  \cite{Sc:04}. This representation is particularly suitable for interconnected and networked systems that include energy networks \cite{GrBeZh:25, HaMaMe:20, StMaCu:22, ShGe:25}, mechanical systems \cite{SaDiSc:25} and thermodynamic processes \cite{RaMeSb:13}.

Beyond modeling and control, pH systems have also attracted considerable attention in network optimization and optimal control problems, where their energy-based structure enables the systematic incorporation of physical constraints into optimization formulations, see for instance \cite{DoKlLaTo:23, PrGeSc:26}. For optimal control problems on pH systems with input $u$ and passive output $y$, a natural candidate for the cost function is the energy $\int_0^T u^\top(t) y(t) \mathrm{d}t$ supplied to the system at time~$t$ integrated on the time horizon $[0,T]$. Such control problems are referred as energy-optimal control problem (EOCP). For instance, given state $x_c(t)\in\rline^n$, input $u_c(t)\in\rline$, output $y_c(t)\in\rline$ and Hamiltonian $H(x)$, pointwise skew-symmetric $J(x)\in\rline^{n\times n}$ and pointwise symmetric positive semi-definite\footnote{The set of real numbers is $\rline$, and $\rline^n$ is $n$-dimensional column vector with all its elements in $\rline$. The space of measurable and absolutely integrable functions is denoted by  $L^1(0,T;\rline)$.} $R(x)\in\rline^{n\times n}$, a finite-horizon energy-optimal output stabilization problem for continuous-time nonlinear pH system has the form 
\begin{align}\label{eq:ct_ocp}\tag{CT-EOCP}
&\min_{u_c \in L^1(0,T;\uline)}\m\m \int\nolimits_0^T u_c(t) y_c(t) + {|c^\top x_c(t)-y_{ref}|^2} \mathrm{d}t \\[0.15em]
&\mathrm{s.t.}\, \ \dot x_c(t) = (J(x_c(t))-R(x_c(t)))Qx_c(t) + bu_c(t)\m, \quad x_c(0) = x_0\m \nonumber\\[0.15em]
&\mathrm{and}\, \  y_c(t)=b^\top Qx_c(t), \ \ u_c(t)\in\uline \qquad \forall t\in [0,T]. \nonumber
\end{align}
In \cite{PhScWo:24}, a thermodynamically consistent version of \eqref{eq:ct_ocp} has been discussed under the assumption that optimal trajectories are uniformly bounded in $T$ and it has been established that the open-loop optimal solutions of \eqref{eq:ct_ocp} spend most of their time around the set $\yline_{ct}= \zline_{ct}\cap\xline_r$, where $\xline_r = \{z\in\rline^n\mid c^\top z=y_{ref}\}$ and $\zline_{ct} = \{z\in\rline^n\mid R^{\nicefrac{1}{2}}(z)Qz = 0\}$\footnote{The Cholesky factor of $R$ is denoted by $R^{\tfrac{1}{2}}$.} is the conservative manifold of the pH system. This behavior of open-loop optimal solution of \eqref{eq:ct_ocp}, also termed as turnpike behavior, follows from the strict dissipativity w.r.t manifold property of \eqref{eq:ct_ocp} which is described by the dissipation inequality
\begin{equation*}\label{eq:diss_ineq_cont}
    H(x_c(T)) - H(x_c(0)) \leq \int_0^T u_c(t)y_c(t) - \alpha_l(\dist(x_c(t),\yline_{ct}))\operatorname{d}t  
\end{equation*}
for some $\alpha_l\in\Kscr_\infty$\footnote{A continuous function $\alpha:[0,\infty) \to [0,\infty)$ is said to be of class $\mathcal{K}_\infty$ if $\alpha(0)=0$, $\alpha$ is strictly increasing and unbounded. For a set $\yline$, denote $\dist(x,\yline) = \inf_{w\in\yline} \|x-w\|$.}. The set $\yline_{ct}$ is termed as turnpike set. An energy-optimal set-point tracking version of \eqref{eq:ct_ocp} for linear pH systems is presented in \cite{Schaller21}, for linear descriptor pH system is addressed in \cite{FaMaPh:22}, for nonlinear descriptor pH systems is discussed in \cite{Ka:24} and for infinite-dimensional pH system is in \cite{PhScFa:21}.  All of the aforementioned works exploit the underlying pH structure to establish strict dissipativity of the OCP, and hence obtain the turnpike behavior of open-loop optimal solutions, w.r.t.\ a conservative subspace/manifold. Note that the structure $u^\top y$ in the cost function of an energy-optimal control problem poses significant challenges \cite{FaKiMe:23}, however its advantages over standard cost functions such as $u^\top u$ have been studied numerically and presented in \cite{SeScWo:23, ScZeWo:24}. 

Motivated by these works, in \cite{SaSiSc:25} we investigated the EOCP for a discrete-time nonlinear pH system and utilized a difference and differential representation \cite{Monaco1998, Moreschini2019} to retrieve the strict dissipativity of the EOCP and turnpike property of open-loop optimal solutions in the discrete-time setting. Retrieving the turnpike property of open-loop optimal solutions for EOCP is particularly important 
to utilize optimization-based control such as (economic) model predictive control \cite{book:Grune2017}, where one iteratively solves a finite-horizon  optimal control problem to define a nonlinear feedback law and the resulting closed-loop trajectory is obtained by applying a portion of the computed open-loop optimal input sequence to the system.
While the turnpike property of open-loop optimal solutions has been discussed in aforementioned works on energy-{optimal} OCPs, the MPC closed-loop behavior 
has never been investigated before. In this work, we bridge this gap by embedding the EOCP within a receding-horizon framework and analyzing the  performance of the resulting closed-loop system. In particular, we utilize ideas from the field of {dissipativity-based model predictive control (MPC)}, also known as economic MPC \cite{book:Grune2017}. 

Dissipativity-based MPC, in the absence of terminal ingredients, usually relies on strict dissipativity of the optimal control problem and on appropriate reachability assumptions to establish turnpike behavior with respect to an optimal steady-state i.e. to ensure that the open-loop optimal solutions remain in a neighborhood of the optimal steady-state for majority of the time-horizon \cite{FaulBonv15,GrMu:16,FaGrMu:18}. Extension to problems when the turnpike object is a periodic orbit instead of a steady-state are discussed in \cite{GrSt:14, Mu:21, MuGr:16, ZaGrDi:17} and to problems where turnpike object is a general control-invariant set are in \cite{Gr:13, DoAn:18}. To utilize the turnpike properties of the open-loop optimal solutions for obtaining performance guarantees on the closed-loop solutions, the forward control invariance property of the turnpike object plays an essential role in all these aforementioned works. For EOCP of discrete-time pH systems considered in this work, the turnpike object is a set constructed using the conservative subspace/manifold of the system, which is available directly from the dynamics of the system and the cost function of the optimal control problem. Hence, unlike \cite[Section V]{SaSiSc:25}, we do not rely on solution of a steady-state optimization problem for defining the turnpike object. However, such a turnpike object does not have forward control invariance property, which makes it difficult to comment on the nature of the closed-loop solution obtained by applying dissipativity-based MPC algorithm. {Note that, to the best of our knowledge, this is first work to establish stability properties of the closed-loop optimal solution of energy-optimal control of pH systems.} 

In this work, we consider the energy-optimal output stabilization of a discrete-time nonlinear single-input single-output (SISO) pH system and establish practical stability of the closed-loop optimal trajectories, obtained via a receding horizon dissipativity-based MPC algorithm, to a set $\yline$. The construction of the set~$\yline$ stems naturally from the structure of the underlying pH system and the tracking output stabilization objective and hence does not require the knowledge of steady-states. We obtain the stability guarantees without imposing the restriction that the set $\yline$ has to be forward control-invariant. This generalizes the existing results on stability guarantees for dissipativity-based MPC algorithms such as \cite[Theorem 4.2]{Gr:13}, where the set $\yline$ is an optimal steady-state \cite[Section 5]{Gr:13}, an optimal periodic orbit \cite{MuGr:16} or a general forward control invariant set \cite{DoAn:18}. To overcome the challenge posed by our set $\yline$, we use strict dissipativity of {our OCP} with respect to this set and establish a measure turnpike property that is weaker than the standard notion of measure turnpike used in MPC literature \cite[Chapter 8]{book:Grune2017}. Note that the pH structure also plays important role in establishing recursive feasibility of the EMPC algorithm. The rest of the paper is organized as follows: In Section \ref{se:prelims} we summarize the {difference and differential (DD) representation} for discrete-time pH systems with quadratic Hamiltonian and we formulate the problem addressed in our work. In Section \ref{se:pract_conver} we present our results on recursive feasibility, the measure turnpike property and the practical stability of the closed-loop system. We illustrate our results using a linear and a nonlinear example in Section \ref{se:Mot_num_Ex}. We draw conclusions in Section \ref{se:concl}.

{\bf Notation:} The set of natural numbers is $\nline$. Denote $\nline_0=\nline\cup\{0\}$ and given $a,b\in\nline_0$ such that $a\leq b$, denote $[a:b]=[a,b]\cap\nline_0$. The set of real numbers is $\rline$ and $\rline_{\geq 0}=\{a\in\rline\mid a\geq 0\}$. A continuous function $\alpha:\mathbb{R}_{\geq 0} \to \mathbb{R}_{\geq 0}$ is said to be of class $\Kscr$ if $\alpha(0)=0$ and $\alpha$ is strictly increasing. If, in addition, $\alpha$ is unbounded, then $\alpha$ is said be of class $\Kscr_\infty$. We say $\sigma\in\Lscr_\nline$ if $\sigma:\nline\to\rline_{\geq 0}$ is decreasing with $\lim_{k\to\infty}\sigma(k)=0$. For a set $\yline$, $|\yline|$ denotes its cardinality and  $\dist(x,\yline) = \inf_{w\in\yline} \|x-w\|$.  Let $\operatorname{int} \yline$ define the interior of set $\yline$. Identity matrix and zero matrix of dimension $n\times n$ are denoted by $I_n$ and $\mathbf{0_n}$, respectively, and $0_n$ is a column vector with all its element zero.   \vspace{-3mm}
\section{Preliminaries and problem description}\label{se:prelims}\vspace{-2mm}
\ \ \ In this section, we recall the DD representation of discrete-time nonlinear systems from \cite[Chapter 4]{MoThesis:21}, discuss the notion of strict dissipativity with respect to a set and state our problem formulations. In our previous work \cite{SaSiSc:25}, this particular representation was crucial in retrieving certain results of \cite{Schaller21} and \cite{PhScWo:24} for the discrete-time pH setting. While the DD representation has been discussed for general non-linear pH systems, we restrict the discussion to pH system with quadratic Hamiltonian as that is sufficient for problem addressed in this work. We direct the reader to \cite{Monaco1998}, \cite{Moreschini2019}, \cite{moreschini2023dirac} and to the references therein for detailed discussion on DD representation of pH systems. 

\emph{Difference and differential representation:} Let $H:\rline^n\to[0,\infty)$ be such that $H(x)=\tfrac{1}{2}x^\top Qx$ and $Q\in\rline^{n\times n}$ is positive-definite. Given the Hamiltonian $H$, the difference and differential representation of the dynamics of a discrete-time single-input single-output (SISO) nonlinear pH system is  
\begin{align}
x_{k+1} &= x_k + \tfrac{1}{2}(J(x_k)-R(x_k))Q(x_k+x_k^+) + bu_k,  \label{eq:ddr_state} \\[0.14em]
y_k &= \tfrac{1}{2}b^\top Q(x_k^+ +x_{k+1}) \qquad \forall\m k\in\nline_0, \label{eq:ddr_output}
\end{align}
where $x_k\in\rline^n$ is the state, $u_k\in\rline$ is the input, $y_k\in\rline$ is the output, $J(x_k)\in\rline^{n\times n}$ is pointwise skew-symmetric, $R(x_k)=R^\top(x_k)\in\rline^{n\times n}$ is pointwise positive semi-definite, $b\in\rline^n$,
\begin{equation}\label{eq:phs_ddr_interm}
x_k^+ = x_k + \tfrac{1}{2}(J(x_k)-R(x_k))Q(x_k+x_k^+) \qquad\forall\m k\in\nline_0
\end{equation}
and $x_0\in\rline^n$ is the initial condition. By computing explicit expression $x_k^+$ in terms of $x_k$ from \eqref{eq:phs_ddr_interm} and using it to rewrite \eqref{eq:ddr_state}-\eqref{eq:ddr_output}, the dynamics of the pH-system is expressed in an explicit form 
\begin{align}
x_{k+1} &= J_1^{-1}(x_k)J_2(x_k)\m x_k+ bu_k \m, \label{eq:phs_ddr_state} \\
y_k&= b^\top Q J_1^{-1}(x_k)J_2(x_k)\m x_k + \tfrac{1}{2}b^\top Q bu_k \label{eq:phs_ddr_output}
\end{align}
for each $k\in\nline_0$ and initial condition $x_0$, where for any $x\in\rline^n$
$$J_1 (x) = I - \tfrac{1}{2}(J(x)-R(x))Q, \quad J_2 (x) = I + \tfrac{1}{2}(J(x)-R(x))Q.$$
The invertibility of $J_1(x)$ for positive semi-definite $R(x)$ and $Q$ is established in \cite[Lemma 1]{SaSiSc:25}. Henceforth,  we refer to the system described by \eqref{eq:phs_ddr_state}-\eqref{eq:phs_ddr_output} as the considered discrete-time pH system. Note that an explicit representation of the pH system is possible because of the quadratic Hamiltonian function $H$, which also allows us to directly compute 
\begin{equation}\label{eq:disc_ham_ineq}
    H(x_{k+1})-H(x_k) = u_ky_k -\|R^{\nicefrac{1}{2}}(x_k)QJ_1^{-1}(x_k)x_k\|^2_2.
\end{equation}
The aforementioned equation induces an energy balance, which is an inherent characteristics of pH systems \cite{SchaJelt14}. Indeed, from \eqref{eq:disc_ham_ineq} we have $H(x_{k+1})-H(x_k) \leq u_ky_k$, where the equality holds for any $x_k\in\rline^n$ if and only if $R^{\nicefrac{1}{2}}(x_k)QJ_1^{-1}(x_k)x_k=0$. Hence, interpreting $u_ky_k$ as the supplied energy, \eqref{eq:disc_ham_ineq} dictates that the change in energy (represented by $H$) of the pH-system can never exceed the supplied energy.
Define 
\begin{equation}\label{eq:cons_mani}
\zline = \{z\in\rline^n\m |\m R^{\nicefrac{1}{2}}(z)QJ_1^{-1}(z)z=0\}
\end{equation}
and note that $\zline$ is a manifold on which the non-negative dissipation term on the right-hand side of \eqref{eq:disc_ham_ineq} vanishes and the Hamiltonian of the pH system is preserved. Hence, we refer to $\zline$ as the conservative manifold of the pH system.
{\begin{remark}[Nonlinear descriptor form]\label{re:descriptor}
Let $z_k \coloneqq x_k^+$ for each $k\in\nline_0$. From \eqref{eq:phs_ddr_interm}, we have $z_k = J_1^{-1}(x_k)J_2(x_k)x_k$. Using $z_k$,  \eqref{eq:ddr_state}-\eqref{eq:phs_ddr_interm} can be transformed into a nonlinear quasi-Weierstrass descriptor form 
\begin{align*}
 \bbm{I_n & \mathbf{0_n}\\ \mathbf{0_n} & \mathbf{0_n}}\bbm{x_{k+1}\\z_{k+1}} &= \bbm{\mathbf{0_n}& I_n\\ \mathbf{0_n} & -I_n}\bbm{x_k\\ z_k} + \bbm{0_n\\ J_1^{-1}(x_k)J_2(x_k)x_k} + \bbm{b\\0_n}u_k,  \\[0.1em]
 y_k &= \bbm{0_n & b^\top Q}\bbm{x_k\\z_k} + \tfrac{1}{2}b^\top Q b\m u_k, \qquad \forall\m k\in\nline_0
\end{align*}
with $x_0\in\rline^n$ and $z_0 = J_1^{-1}(x_0)J_2(x_0)x_0\in\rline^n$. An example of a linear descriptor system in quasi-Weierstrass form is given in \cite[Eqs.~5(a)-5(b)]{ScFaWo:22}. 
\end{remark}}

\emph{Strict dissipativity of a discrete-time OCP:} We recall the notion of strict dissipativity w.r.t to a manifold $\yline$ from \cite[Definition 1]{SaSiSc:25}. Let $\xline\subset \rline^n$ and $\uline\subset \rline^m$ be the set of admissible state and input values. For $N\geq 1$ and initial state $x(0)=x_0\in \xline$, denote $\uline^N(x_0)= \{(u_k)_{k=0}^{N-1} \subset \uline \mid x(k,x_0)\in\xline, k \in [1:N] \}$, where $x(k, x_0)$ is the state at time $k \in [0:N]$ subject to $x(0,x_0) = x_0$ and application of $u=(u_k)_{k=0}^{N-1}$.
\begin{definition}[Strict dissipativity of discrete-time OCPs]\label{de:strict_diss_disc} 
For $x_0\in \mathbb{X}$, the discrete-time OCP 
\begin{align*}
\begin{split}
\min_{u\in\uline^N(x_0)}&\sum\nolimits_{k=0}^{N-1} l(x_k,u_k) \quad \mathrm{s.t.} \ \ x_{k+1} = f(x_k,u_k),\quad {y_k = g(x_k,u_k),} \\ 
&\mathrm{and}\, \ x(0,x_0) = x_0, \quad x(k,x_0)\in\xline \qquad \forall\m k\in[1:N].
\end{split}
\end{align*}
is said to be strictly dissipative w.r.t.\ a set $\yline\subset\rline^n$ if there exists storage function $S : \xline \to [0,\infty)$ and $\alpha\in\Kscr_{\infty}$ such that for all optimal control sequences $u^*\coloneqq (u_k^*)_{k=0}^{N-1}$ and associated optimal state $x^* \coloneqq (x^*_k)_{k=0}^{N}$, the dissipation inequality 
 \begin{equation}\label{eq:str_disp_subsp_K}
     S(x^*_{k+1})-S(x^*_k) \leq l(x_k^*,u_k^*) - \alpha(\dist(x^*_k,\yline))
 \end{equation}
 holds for each $k\in [0:N-1]$. Note that \eqref{eq:str_disp_subsp_K} leads to     $ S(x_N)-S(x_0) \leq \sum\nolimits_{k=0}^{N-1} l(x_k^*,u_k^*) - \alpha(\dist(x^*_k,\yline))$,
when summed over $k\in[0:N-1]$.
\end{definition}
{\it Problem description:} For the system in \eqref{eq:phs_ddr_state}-\eqref{eq:phs_ddr_output}, consider the finite-horizon energy-optimal output stabilization problem
\begin{align}
&\min_{u \in \uline^N(x_0)}\m \JJJ_N(x_0,u) \m =\m \min_{u_k \in \uline(x_0)}\m\m \sum_{k=0}^{N-1}  u_k y_k+ {|c^\top x_k-y_{ref}|^2} \nonumber\\
&\mathrm{s.t.}\ \ x_{k+1} = J_1^{-1}(x_k)J_2(x_k)\m x_k+ bu_k \m, \quad x_k\in\xline_C  \label{eq:e_ocp}\tag{E-OCP} \\[0.15em]
&\mathrm{and} \ \ y_k= b^\top Q J_1^{-1}(x_k)J_2(x_k)\m x_k + \tfrac{1}{2}b^\top Q\m bu_k, \ \ \forall k\in[0:N-1] \nonumber.
\end{align}
Here $x_0$ is the initial condition, {$\uline \subseteq\rline$ is a compact, convex control constraint set with  $0\in\operatorname{int}\uline$}, $c\in\rline^n$, $y_{ref}\in\rline$ is a constant reference, and 
\begin{equation}\label{eq:sublevel_set}
\xline_C = \{z\in\rline^n\m|\m H(z) = \tfrac{1}{2}z^\top Qz\leq C \}
\end{equation}
for some positive scalar $C>0$. 
\begin{remark}\label{re:xline_c}
The result in this work hold for arbitrarily large but finite choice of $C>0$ as long as the \eqref{eq:e_ocp} remains feasible. The requirement that $C>0$ is finite ensures that $\xline_C$ in \eqref{eq:sublevel_set} is compact, which is useful for defining our turnpike set $\yline$, see \eqref{eq:yline} and the discussion below Assumption \ref{as:feasiblity_track_ddr}.
\end{remark}
In the context of \eqref{eq:e_ocp}, we define the following notations and terminology: For $k\in\nline$, a control input $u_k$ is said to be an admissible input if $u_k\in\uline$ and state $x_{k+1}$ obtained by applying this input lies in $\xline_C$.
For $N\in\nline$, a control sequence $\{u_k\}_{k=0}^{N-1}$ is said to be admissible if $u_k\in\uline$ for all $k=[0:N-1]$ and $x_u(k,x_0)\in\xline_C$ for each $k=[0:N]$, where $x_u(k, x_0)$ is the state at time-instant $k$ obtained by solving \eqref{eq:phs_ddr_state} via control sequence $u_0,u_1,\ldots, u_{k-1}$. Whenever there is no ambiguity, we denote $x_k\coloneqq x_u(k, x_0)$. For any $N\in\nline$ and $x_0\in\xline_C$, define the optimal value function for \eqref{eq:e_ocp} as
\begin{equation}\label{eq:opti_val_func}
    V_N(x_0) = \inf_{u\in\uline^N(x_0)} \JJJ_N(x_0, u).
\end{equation}
To ensure the existence of optimal input sequence for \eqref{eq:e_ocp}, we impose the assumption below.
\begin{assumption}\label{as:exist_opti}
For each $x_0\in\xline_C$, there exists an optimal input sequence $u^*_{N,x_0}\in\uline^N(x_0)$ such that $V_N(x_0) = \JJJ_N(x_0, u^*_{N,x_0})$ holds for \eqref{eq:e_ocp}.
\end{assumption}

Next, we present the closed-loop system that is obtained by applying a receding horizon control algorithm and then state the objective of this work. For $N\in\nline$, let $\mu_N:\rline^n\to\rline$ denote a feedback control law. The dynamics of the resulting closed-loop system is
\begin{equation}\label{eq:cl_MPC}
    x^{cl}_{k+1} = J_1^{-1}(x^{cl}_k)J_2(x^{cl}_k)x^{cl}_k+ b\mu_N(x^{cl}_k) \quad \forall\m k\in[0:K-1]
\end{equation}
with $x^{cl}_0 = x_0$, where $x^{cl}_k$ is the closed-loop state obtained by applying the feedback control at time-step $k-1$ and $K\in\nline$. For horizon length $K\in\nline$ and initial condition $x_0\in\xline_C$, denote the closed-loop cost function by $\JJJ^{cl}_K(x_0, \mu_N)$. To identify $\mu_N$ and corresponding closed-loop state trajectory $x^{cl}$, we fix an optimization horizon $N\in\nline$ and employ the following receding horizon MPC algorithm: (i) Measure the current state $x^{cl}_{k}$ of the system for $k\in[0:K-1]$. (ii) With initial condition $x_k^{cl}$, solve for control action $u\in\uline^N(x^{cl}_k)$ that minimizes $\JJJ_N(x^{cl}_k,u)$ subject to the constraint in \eqref{eq:e_ocp} and denote the corresponding optimal solution by $u^*_{N,x^{cl}_k}$. (iii) Define the feedback control value $\mu_N(x_k^{cl}) = u_{N,x_k^{cl}}^*(0)$ and apply $\mu_N(x_k^{cl})$ to obtain the next state $x_{k+1}^{cl}$. (iv) Repeat Steps (i)-(iii) until $k=K-1$. The goal of this work is to establish the practical stability of the closed-loop state trajectory $x_k^{cl}$ obtained by applying the receding horizon controller $\mu_N$ in \eqref{eq:cl_MPC} to solve \eqref{eq:e_ocp}. 
\begin{remark}[Extension to multi-input multi-output setting]\label{re:mimo}
Since we consider a constant input map $b\in\rline^n$, an extension of the DD representation to multiple inputs is expected to be relatively straightforward. We nevertheless restrict our attention to the single-input case in order to focus on the main objective of this work, namely, closed-loop analysis of EOCP. DD representation of a two-input system with nonlinear input map has been discussed in \cite{Monaco_2011}, see also \cite[Remark 3.4]{Moreschini2019}, \cite[Remark 2.1]{moreschini2023dirac}.
\end{remark}
\begin{remark}[Discretization frameworks]\label{re:siso_ddr}
Apart from DD representation, there exist several other discrete-time formulations of pH systems including formulations based on discrete-gradient \cite{KiMoSc:25} and symplectic integration \cite{kotyczka2019, KoTh:21}. While these methods are tailored towards preserving passivity and designing stabilizing controllers \cite{Ma:23, MoKoLe:24}, employing them for obtaining the strict dissipativity of the EOCP in discrete-setting is not straightforward. For discrete-gradient type methods, this has been highlighted by a simple example in \cite[Section 3]{SaSiSc:25}.   
\end{remark}
\section{Closed-loop practical stability} \label{se:pract_conver}
\ \ \ Consider the closed-loop optimal state trajectory $x^{cl}=\{x^{cl}_0, x^{cl}_1, \ldots, x^{ck}_N\}$ obtained from \eqref{eq:cl_MPC} by applying the receding horizon MPC algorithm, that is discussed above Remark \ref{re:mimo}, to solve \eqref{eq:e_ocp}. In this section we will establish that for sufficiently large $N$, the trajectory $x^{cl}$ is practically stable with respect to a set $\yline$ (defined in \eqref{eq:yline}) i.e., the trajectory $x^{cl}$ spends most of its time near $\yline$ in an average sense. We begin by establishing that \eqref{eq:e_ocp} remains feasible at each step of the receding horizon algorithm. We construct the turnpike set $\yline$ and prove the strict dissipativity of \eqref{eq:e_ocp} w.r.t $\yline$. Using the strict dissipativity, we establish measure turnpike property of the open-loop optimal solution. Combining all these ingredients, we state and prove the main result of our work in  Theorem \ref{th:Gr13_Th4.2}. We denote $x^*_{i|k} \coloneqq x_{u^*}(i, k)$ to be the optimal state at $i^{\rm th}$ instant which is obtained by applying open-loop optimal input sequence $u^*_{0|k}, u^*_{1|k}, u^*_{2|k}, \ldots, u^*_{i-1|k}$ starting with initial state $x^*_{0|k}$. Note that the closed-loop input obtained via the receding horizon algorithm at time-instant $k$ is $u_k^{cl} = u^*_{0|k}$ and hence $x_{k+1}^{cl} = x^*_{1|k}$. 
\begin{proposition}[Recursive feasibility]\label{pr:rec_fea}
Let Assumptions \ref{as:exist_opti} hold. Then, \eqref{eq:e_ocp} is recursively feasible for any initial condition $x\in\xline_C$. 
\end{proposition}
\begin{proof}
For the pH system in \eqref{eq:e_ocp}, the set $\xline_C$ is controlled forward invariant i.e. for any $x_k\in\xline_C$ with $k\in\nline_0$ there exists a $u_k\in\uline$ such that $x_{k+1}\in\xline_C$. This is seen from the dissipation inequality \eqref{eq:disc_ham_ineq} by letting  $x_k\in\xline_C$ and $u_k=0\in\operatorname{int}\uline$ which leads to $H(x_{k+1})\leq H(x_k)\leq C$ for any $k\in\nline_0$. Hence $x_{k+1}\in\xline_C$.

Suppose that \eqref{eq:e_ocp} is feasible at the closed-loop state $x^{cl}_k$ and let $u^*_{0|k}, u^*_{1|k}, u^*_{2|k}, \ldots, u^*_{N-1|k}$ be the open-loop optimal input sequence for the initial condition $x^{cl}_k$. Owing to the feasibility of \eqref{eq:e_ocp} at $x^{cl}_k$, we have $x^*_{i|k}\in\xline_C$ for each $i\in[0:N]$, $u^*_{i|k}\in\uline$ for each $i\in[0:N-1]$ and so  $x^{cl}_{k+1}=x^*_{1|k}\in\xline_C$. Next, with initial state $x^{cl}_{k+1}$ consider the control sequence $\hat u_{i|k+1} = u^*_{i+1|k}$ if $i\in[0:N-2]$ and $\hat u_{N-1|k} = 0$. Clearly, the input $\hat u_{i|k+1}\in\uline$ for each $i\in[0:N-2]$ and the corresponding state $\hat x_{i|k+1}=x^*_{i+1|k}\in\xline_C$ for each $i\in[0:N-1]$. Lastly, since $\xline_C$ is invariant with $u=0\in\operatorname{int}\uline$, $\hat x_{N|k}\in\xline_C$. Hence, the optimal control problem \eqref{eq:e_ocp} is feasible at the initial condition $x^{cl}_{k+1}$ and recursively feasible for each $k\in \nline$ under the receding horizon algorithm described in Section \ref{se:prelims}.  
\end{proof}
Using $\zline$ defined in \eqref{eq:cons_mani}, $\xline_C$ in \eqref{eq:sublevel_set} and the set 
\begin{equation*}\label{eq:xline_r_set}
\xline_r = \{z\in\rline^n\m|\m c^\top z = y_{ref}\},
\end{equation*}
we define the turnpike set to be
\begin{equation}\label{eq:yline}
\yline = \zline\cap\xline_r\cap\xline_C.
\end{equation}
To establish the strict dissipativity of  \eqref{eq:e_ocp} with respect to the set $\yline$ in \eqref{eq:yline}, we impose the following assumption.
\begin{assumption}\label{as:feasiblity_track_ddr}
The map $x\to R^{\nicefrac{1}{2}}(x)QJ_1^{-1}(x)x$ is continuous and the turnpike set $\yline$ is nonempty. 
\end{assumption}
Assumption \ref{as:feasiblity_track_ddr} ensures that $\yline$ is compact. Indeed, by definition $\xline_C$ is a compact set and $\xline_r$ is closed. Continuity of map $x\to R^{\nicefrac{1}{2}}(x)QJ_1^{-1}(x)x$ ensures that the conservative manifold $\zline$ is also closed. Hence compactness of $\yline$ follows from the structure of $\yline$. Furthermore, via Assumption \ref{as:feasiblity_track_ddr}, there exists $x\in\xline_C$ such that $c^\top x=y_{ref}$ and $R^{\nicefrac{1}{2}}(x)QJ_1^{-1}(x)=0$. Next, we present a lemma that will be used to establish strict dissipativity of \eqref{eq:e_ocp} w.r.t $\yline$.
\begin{lemma}\label{le:sdp_Y}
Suppose Assumption \ref{as:feasiblity_track_ddr} holds. Then, there exists $\alpha_l \in \Kscr_\infty$ such that
\begin{equation}\label{eq:sdp_stage_cost}
uy + |c^\top x-y_{ref}|^2 \m\geq\m H(f(x,u)) - H(x) + \alpha_l(\dist(x, \yline)) 
\end{equation}
holds for any $x\in\xline_C$ and any $u\in\uline$, where  $f(x,u)= J_1^{-1}(x)J_2(x)\m x+ bu$.    
\end{lemma}
\begin{proof}
Let $x\in\xline_C$, $u\in\uline$ and $\BBB(x) =  \|R^{\nicefrac{1}{2}}(x)QJ_1^{-1}(x)x\|^2 + |c^\top x- y_{ref}|^2$. Using $\BBB(x)$ and \eqref{eq:disc_ham_ineq}, the left-hand side of \eqref{eq:sdp_stage_cost} is written as
\begin{align*}
u^\top y + |c^\top x-y_{ref}| = H(f(x,u))-H(x) + \BBB(x)  
\end{align*}
We will establish \eqref{eq:sdp_stage_cost} by proving that there exists a  class $\Kscr_{\infty}$ function $\alpha_l$ such that
\begin{equation}\label{eq:BBB_x}
\BBB(x)\geq \alpha_l(\dist(x,\yline)),\qquad \forall x\in\xline_C.    
\end{equation}
Note that the map $\BBB:\xline_C\to\rline_{\geq 0}$ is positive definite  on $\xline_C$ with respect to the set $\yline$ because $B(x)\geq 0$ for all $x\in\xline_C$ and $B(x)=0$ if and only if $x\in\yline$. For $r\geq 0$, define 
$$m(r) = \min_{\substack{x\in\xline_C, \\ \dist(x,\yline)\geq r}} \BBB(x).$$
Using Assumption \ref{as:feasiblity_track_ddr}, we get that $m(r)$ is well-defined since $\xline_C$, $\yline$, are nonempty compact sets and $B$ is a continuous map. In addition, (i) $m(0)=0$, (ii) $m(r)>0$ for every $r>0$ and (iii) $m$ is a non-decreasing function of $r$. Fix $r_{\max} = \max_{x\in\xline_C} \dist(x, \yline)$,
which is well-defined. Select a class $\Kscr$ function $\bar\alpha_l$ such that 
$\bar\alpha_l(r) \leq m(r)$ for all $r\in\m[0,\m r_{\max}]$.
For $r>r_{max}$, extend $\bar\alpha_{l}$ with an arbitrary class $\Kscr_{\infty}$ function to obtain $\alpha_l\in\Kscr_\infty$. Then, \eqref{eq:BBB_x} holds which further implies \eqref{eq:sdp_stage_cost}.
\end{proof}

The strict dissipation inequality in \eqref{eq:sdp_stage_cost} holds for all admissible input $u$ and admissible state $x$ and hence it also holds for the optimal input $u^*$ and associated optimal state $x^*$. Therefore, we naturally have strict dissipativity of (EOCP) w.r.t the turnpike set $\yline$ as per Definition \ref{de:strict_diss_disc} with stage cost
\begin{align*}
l(&x,u) = uy + |c^\top x_k-y_{ref}|^2 \\
&= H(f(x,u))-H(x) + \|R^{\tfrac{1}{2}}(x)QJ_1^{-1}(x)x\|^2 +|c^\top x_k-y_{ref}|^2,
\end{align*}
and $S = H$ in that definition. This can be checked by replacing $x$, $u$, $f(x,u)$ with $x^*_k$, $u_k^*$ and $x^*_{k+1}$, respectively, in \eqref{eq:sdp_stage_cost} and summing \eqref{eq:sdp_stage_cost} over $k\in [0:N-1]$. We remark that the result in Lemma \ref{le:sdp_Y} utilizes energy-balance equation in \eqref{eq:disc_ham_ineq}. For the rest of the discussion in this section, instead of working with $l(x,u)$,
we work with rotated stage cost $\tilde l(x,u) = l(x,u)- H(f(x,u))+H(x)$. Using Lemma \ref{le:sdp_Y}, we have
\begin{equation}\label{eq:tilde_l}
    \tilde l(x,u) = \|R^{\nicefrac{1}{2}}(x)QJ_1^{-1}(x)x\|^2_2 + {|c^\top x-y_{ref}|^2} \geq \alpha_l(\dist(x,\yline)).
\end{equation}
Using \eqref{eq:opti_val_func}, the rotated optimal value function is $\tilde V_N(x) = V_N(x) -H(f(x,u))+H(x)$ and it rewritten as
\begin{align}\label{eq:frodo}
\tilde V_N(x) = \sum_{k=0}^{N-1}\tilde l(x_{u^*}(k,x), u_k^*) \geq \sum_{k=0}^{N-1}\alpha_l(\dist(x_{u^*}(k,x), \yline)) 
\end{align}
for all $x\in\xline_c$ and for all $N\in\nline$. The associated optimal control problem is obtained by replacing $\JJJ_N$ and $l$  with $\tilde \JJJ_N$ and $\tilde l$, respectively. Note that the inequality in \eqref{eq:tilde_l}-\eqref{eq:frodo} is a result of the strict dissipativity with respect to $\yline$ property established in Lemma \ref{le:sdp_Y}. Furthermore, recursive feasibility in Proposition \ref{pr:rec_fea} remains unaffected if we replace $\JJJ_N$ in terms of $l(x,u)$ with $\tilde\JJJ_N$ in terms of $\tilde l(x,u)$. 
\begin{assumption}[Reachability of turnpike set $\yline$]\label{as:reach_Y}
There exists $M_0\in\nline$ such that for each $w\in\xline_C$ there is a $z\in\yline$ and an admissible control sequence $v\in\uline^{K_w}(w)$, such that $x_{v}(K_w,w) = z$ for $K_w\leq M_0$, where $z$, $v$ and $K_w$ depend on the initial state $w$.
\end{assumption}
Assumption \ref{as:reach_Y} ensures that the turnpike set $\yline$ is reachable from any state $x\in\xline_C$ in a finite number of time steps. This is not restrictive since $\xline_c$ is a compact set and the optimal control problem is recursively feasible inside $\xline_C$, see also Assumption 3.10 and Remark 3.12 in \cite{PhScWo:24}. However, unlike \cite{PhScWo:24} and most of the existing results on measure turnpike behavior via the strict dissipativity property, we do not have access to a controlled optimal steady-state, optimal periodic orbit or a general control-invariant set. Hence, one cannot ensure that a trajectory reaching some state inside $\yline$ can be forced to remain inside $\yline$ for the remaining part of the horizon. This makes it difficult to obtain the standard measure turnpike behavior, which usually requires an upper bound on $\tilde V_N(x)$ that is independent of the horizon length. Such a horizon-independent bound is often established using uniform cost-bounded reachability assumption to the optimal steady-state, optimal periodic orbit or a general control-invariant set, see for instance \cite[Proposition 4.1]{FaGrMu:18}. We handle the aforementioned difficulty using assumption below.
\begin{assumption}[Cheap approximate forward invariance of $\yline$]\label{as:tube_Y}
There exists a monotonically decreasing function $M_1:(0,\infty)\to\nline$ such that for each $z\in\yline$ there exists an admissible control sequence $u_{z,\delta}\in\uline^{L}(z)$  which ensures $x_{u_{z,\delta}}(k,z)\in\Bscr_{\delta}(\yline)$ for each $k\in[1:L-1]$ and  $x_{u_{z,\delta}}(L,z) \in\yline$ for every $L\geq M_1(\delta)$. Furthermore, the cost associated cost with such a state sequence has an upper bound of the form
\begin{equation}\label{eq:smaug}
\sum_{k=0}^{L-1}\tilde l(x_{u_{z,\delta}}(k,z),u_z(k))\leq \beta_0(L) \qquad  \forall L\geq M_1(\delta),
\end{equation} 
where $\beta_0:\nline\to(0,\infty)$ with $\lim_{L\to\infty}\tfrac{\beta_0(L)}{L} = 0$.  
\end{assumption}
Assumption \ref{as:tube_Y} ensures that once the state is initialized on the set $\yline$, it can be steered inside an arbitrarily prescribed $\delta$-neighborhood of $\yline$ for any sufficiently long horizon, to be brought back to $\yline$ at the end of the horizon. The accumulated rotated cost of this maneuver grows sub-linearly with its length. Consequently, the average rotated cost of such a maneuver vanishes as $L\to\infty$. The monotonicity of $M_1$ reflects the fact that smaller neighborhoods of $\yline$ impose a more restrictive requirement and may require longer horizons. If, in addition, the terminal condition in the statement of Assumption \ref{as:tube_Y} is replaced with $x_{u_{z,\delta}}(L,z)=z\in\yline$, then the maneuver is a periodic orbit. In the next result, we establish that for sufficiently large horizon $N\in\nline$, the open-loop optimal state sequence stays in a neighborhood of set $\yline$ for most part of the horizon.
\begin{theorem}[Measure turnpike w.r.t $\yline$]\label{th:meas_turn_Y}
Suppose that Assumption \ref{as:exist_opti} and Assumptions \ref{as:feasiblity_track_ddr}-\ref{as:tube_Y} hold. Let $x\in\xline_C$, $u^*\in\uline^N(x)$ be the open-loop optimal input sequence and $x_k^*\coloneqq x_{u^*}(k,x)$ with $k\in[1:N]$ be the associated open-loop optimal state sequence originating from initial state $x$. Then, for sufficiently large horizon $N\in\nline$, the cardinality\footnote{For a set $\Omega$, we denote $|\Omega|$ to be its cardinality.} of the set $\qline_\delta\coloneqq \{k\in[0:N]\mid \dist(x_k^*,\yline)\leq \delta\}$ satisfies
\begin{equation}\label{eq:meas_tunrpike}
|\qline_\delta| \geq N - \frac{\beta_1(M_0, N)}{\alpha_l(\delta)}
\end{equation}
for each $x\in\xline_C$ and every $u^*\in\uline^N(x)$ and some function $\beta_1$.
\end{theorem}
\begin{proof}
Fix $x\in\xline_C$ and $\delta>0$. For the initial state $x$, consider the rotated optimal value function $\tilde V_N(x)$ defined in \eqref{eq:frodo} with the optimal input sequence $u^*\in\uline^N(x)$. Rewrite the inequality in \eqref{eq:frodo} using the sets $\qline_\delta$ and $[0:N]\setminus\qline_\delta$ as 
\begin{align*}
\tilde V_N(x) &\geq \sum_{k\in\qline_\delta} \alpha_l(\dist(x_k^*,\yline)) + \sum_{k\in [0:N]\setminus\qline_\delta} \alpha_l(\dist(x_k^*,\yline)).
\end{align*}
Since $\alpha_l\in\Kscr_\infty$, we get $\tilde V_N(x)\geq (N-|\qline_\delta|)\alpha_l(\delta)$ from the inequality above. Next, identify $K_x\in\nline$ with $K_x\leq M_0$, $z\in\yline$ and $v\in\uline^{K_x}(x)$ from Assumption \ref{as:reach_Y} for initial condition $x\in\xline_C$. Then, using $\delta>0$ and $z\in\yline$, we obtain $L\geq M_1(\delta)$, $u_{z,\delta}\in\uline^L(z)$ from Assumption \ref{as:tube_Y}. Construct a feasible input sequence $\hat v\in\uline^N(x)$ in the following way: 
\begin{equation*}
    \hat v(k) = \begin{cases} v(k) \text{\quad if\quad} k \in [0:K_x-1] \\
    u_{z,\delta}(k)\text{\quad if \quad} k\in[K_x:K_x+L-1]
    \end{cases}
\end{equation*}
such that $N = K_x+L\geq M_0+M_1(\delta)$. By definition of $\tilde V_N(x)$, we have
\begin{align}\label{eq:cantor}
\tilde V_N(x) \leq \sum_{k=0}^{N-1} \tilde l(x_{\hat v}(k,x),\hat v(k))\leq \tilde C_lK_x + \beta_0(L).
\end{align}
The final inequality in \eqref{eq:cantor} is a result of \eqref{eq:smaug} and the face that we have a finite upper bound $0<\tilde C_l<\infty$ on $\tilde l$ since $\xline_C$ is compact. Using $K_x\leq M_0$ and $L=N-K_x$, we get that $\tilde V_N(x)\leq \tilde C_lM_0 + \beta_0(N-K_x)$. Finally, using the lower bound on $\tilde V_N(x)$ along with the last inequality, $(N-|\qline_\delta|)\alpha_l(\delta) \leq \tilde C_lM_0 + \beta_0(N-K_x)$ implies that
\begin{equation}\label{eq:faramir}
|\qline_\delta| \geq N\left(1 - \frac{\tilde C_lM_0}{N\alpha_l(\delta)} - \frac{\beta_0(N-K_x)}{N\alpha_l(\delta)}\right),
\end{equation}
which is \eqref{eq:meas_tunrpike} with $\beta_1(M_0,N)=\tilde C_lM_0 + \beta_0(N-K_x)$.
\end{proof}
When an open-loop optimal solution exhibits standard measure turnpike behavior \cite{FaulBonv15}, the number of time-steps for which the optimal solution steps out of the $\delta$-neighborhood of the turnpike remains uniformly bounded for any $N\in\nline$. However, the measure turnpike property in Theorem \ref{th:meas_turn_Y} is different from the standard notion in the following way: The total number of time-instants at which the optimal state trajectory moves out of $\delta$-neighborhood of the turnpike can increase with $N$ but for sufficiently large $N$, the fraction of these bad time-instants with respect to $N$ decreases and hence we obtain that the open-loop optimal solution spends most of its time near the turnpike set $\yline$ for most part of the horizon. In the next lemma, which is a direct consequence of Theorem \ref{th:meas_turn_Y}, we show that there exist $M$ consecutive time-instants for which the open-loop optimal state remains in a $\delta$-neighborhood of the turnpike set $\yline$, when the horizon $N\in\nline$ is chosen to sufficiently large.
\begin{lemma}[Approximate interval turnpike w.r.t $\yline$]\label{le:block_turnpike}
Let Assumptions \ref{as:exist_opti} and Assumptions \ref{as:feasiblity_track_ddr}-\ref{as:tube_Y} hold. Given $M\in\nline$ and $\delta>0$, for any $x\in\xline_C$ there exists a $i_0\in [0:N-M]$ and $N_1\in\nline$ such that 
\begin{equation}\label{eq:prop_2(a)}
\dist(x_{u^*}(j,x),\yline) \leq \delta \qquad \forall j\in[i_0:i_0+M]
\end{equation}
holds for every $N\geq N_1$, for any optimal state trajectory $x_{u^*}(\cdot,x)$ with the associated optimal input sequence $u^*\in\uline^N(x)$.
\end{lemma}
\begin{proof}
We will use proof by contradiction. Fix $x\in\xline_C$ and let $u^*\in\uline^N(x)$ with optimal state $x_{u^*}(k,x)=x^*_k$ for each $k\in[0:N-1]$ with $N\geq M_0+M_1(\delta)$, where $M_0$ is obtained from Assumption \ref{as:reach_Y} and $M_1(\delta)$ is obtained from Assumption \ref{as:tube_Y}. Given $M\in\nline$ and $\delta>0$, suppose that there does not exist $M$ consecutive time-steps $i_0, i_1, \ldots, i_{M}$ for any $i_0\in[0:N-M]$ such that \eqref{eq:prop_2(a)} holds. Then, the number of consecutive time-steps $k\in[0:N]$ for which $\dist(x^*_k, \yline)\leq\delta$ is at most $M-1$. Then, $|\qline_\delta|\leq (N-|\qline_\delta|+1)(M-1)$, where $\qline_\delta$ is defined from Theorem \ref{th:meas_turn_Y}. Upon rearranging this inequality, we get $N\leq M(N-|\qline_\delta|)+(M-1)$ which results in 
\begin{equation}\label{eq:gimli}
N\leq \frac{M(\tilde C_l M_0 + \beta_0(N-K_x))}{\alpha_l(\delta)} + M-1.
\end{equation}
The last inequality is a result of 
$$N-|\qline_\delta| \leq \frac{\tilde C_l M_0 + \beta_0(N-K_x)}{\alpha_l(\delta)},$$
obtained by rearranging \eqref{eq:faramir} in the proof of Theorem \ref{th:meas_turn_Y}. We obtain a contradiction from \eqref{eq:gimli}, once we divide both sides of \eqref{eq:gimli} by $N$, let $N\to\infty$ and use the fact that $\beta_0(N-K_x)/N \to 0$ as $N\to\infty$. Hence, for sufficiently large $N$, there must exist at least $M$ consecutive time-instant for which \eqref{eq:prop_2(a)} holds.
\end{proof}
The next assumption ensures existence of an admissible input sequence which allows a cost-controlled state transition between two states inside a $\delta$-neighborhood of $\yline$. 
\begin{assumption}[Cost-controlled excursion around $\yline$]\label{as:cheap_block_turn}
Let $\bar\delta>0$. For any $\delta\in(0,\bar\delta\m]$, any $z\in\Bscr_\delta(\yline)$ and $w\in\Bscr_\delta(\yline)$ there exists an admissible control sequence $\hat u\in\uline^{2M}(z)$, scalar $l_{\bar \delta}\geq 0$, $\rho:[0,\bar \delta]\to\rline_{\geq 0}$  with $\rho(\delta)\geq\delta$, $\alpha_{\tilde l}\in\Kscr_\infty$ and $\beta\in\Lscr_\nline$  such that $x_{\hat u}(2M,z) = w$,
\begin{align}
&\dist(x_{\hat u}(k,z), \yline) \leq \rho(\delta) \qquad \forall k\in [1:2M-1],   \label{eq:prop_1a}\\[0.1em]
&\sum_{k=0}^{2M-1} \tilde l(x_{\hat u}(k,z), \hat u(k)) \leq l_{\bar\delta} + \alpha_{\tilde l}(\rho(\delta)) +\beta(M).\label{eq:prop_1b}
\end{align}
\end{assumption}
We remark that the state sequence originating from $z\in\Bscr_\delta(\yline)$, via the admissible input defined by Assumption \ref{as:cheap_block_turn}, are allowed to move outside the $\delta-$neighborhood of $\yline$ as long as such \eqref{eq:prop_1b} holds for such an excursion. Recall $x^{cl}_k$ to be the closed-loop state trajectory obtained by applying the MPC algorithm described below \eqref{eq:cl_MPC}. The corresponding closed-loop control input applied at $k^{\rm th}$ time-instant is $u^{cl}_k$. We now state the main result of this work: the closed-loop optimal state sequence of \eqref{eq:e_ocp}, obtained via \eqref{eq:cl_MPC} is practically stable with respect to the turnpike set $\yline$ in an averaged sense.
\begin{theorem}\label{th:Gr13_Th4.2}
Suppose Assumptions \ref{as:exist_opti} and Assumptions \ref{as:feasiblity_track_ddr}-\ref{as:cheap_block_turn} hold and let $\bar\delta>0$ be given. Then,
\begin{equation}\label{eq:bound_limsup_JC}
\limsup_{P\to\infty} \frac{1}{P}\sum_{k=0}^{P-1} \alpha_l(\dist((x^{cl}_k, \yline)) \leq l_{\bar\delta} + \alpha_{\tilde l}(\rho(\delta)) + \beta(M).   
\end{equation}
holds for all $x\in\xline_C$ and all $N\in\nline$ sufficiently large.
\end{theorem}
\begin{proof}
Given $\bar\delta>0$, let $N\in\nline$ to be sufficiently large such that the results stated in Lemma \ref{le:block_turnpike} hold for \eqref{eq:e_ocp} with initial condition $x\in\xline_C$. Let $u^*\in\uline^N(x)$ be the open-loop optimal input sequence and $x_{u^*}(\cdot, x)$ be the corresponding open-loop optimal state sequence for which \eqref{eq:prop_2(a)} hold. For ease of notation, we will denote $x_{u^*}(\cdot, x)=x_k^*$, whenever obvious.  To establish \eqref{eq:bound_limsup_JC}, we will show that for all $N\in\nline$ chosen to be sufficiently large 
\begin{equation}\label{eq:gandalf}
\tilde V_N(x_{k+1}^{cl})-\tilde V_N(x_k^{cl}) \m\leq\m -\tilde l(x^{cl}_{k}, u^{cl}_{k}) +\epsilon(M,\delta, l_{\bar\delta})
\end{equation}
holds for each $k\in[0:P-1]$, where $\epsilon\in\Lscr_\nline$. To check that \eqref{eq:gandalf} leads to \eqref{eq:bound_limsup_JC}, sum \eqref{eq:gandalf} over $[0:P-1]$, simplify the right-hand side by performing telescopic sum and rearrange the terms to obtain
\begin{equation}\label{eq:sam}
    \sum_{k=0}^{P-1}\tilde l(x_k^{cl}, u_k^{cl}) \leq \tilde V_N(x) - \tilde V_N(x_P^{cl}) + P\epsilon(M,\delta, l_{\bar\delta}).
\end{equation}
Using the inequality in \eqref{eq:tilde_l} and the inequality in \eqref{eq:frodo} which gives $\tilde V_N(x_P^{cl})\geq 0$, \eqref{eq:sam} is rewritten as
 $$\sum_{k=0}^{P-1} \alpha_l(\dist(x_k^{cl},\yline)) \leq \tilde V_N(x) + P\epsilon(M,\delta, l_{\bar\delta}).$$
Then, on dividing the aforementioned inequality by $P$ and taking the limit supremum over $P$ results in \eqref{eq:bound_limsup_JC} by using the fact that $\tilde V_N(x)$ is finite for any $x\in\xline_c$ and $N\in\nline$.  

Next, we prove the inequality in \eqref{eq:gandalf}. For each $k\in[0:P-1]$ and $i\in[0:N]$, recall the notation described above Proposition \ref{pr:rec_fea}. Note that the closed-loop input applied at time-instant $k$ is $u_k^{cl} = u^*_{0|k}$ and hence $x_{k+1}^{cl} = x^*_{1|k}$. Using Lemma \ref{le:block_turnpike}, there exists a set of time-instants $[i_0:i_{M}]$ with $i_0>0$ and $i_{M}<N$ such that  
$$\alpha_l(\dist(x^*_{i|k}, \yline)) \leq \bar\delta \qquad \forall i\in[i_0:i_{M}].$$
Using the above inequality, we write 
\begin{align}
\tilde V_N&(x_k^{cl}) = -\tilde l(x_k^{cl}, u_k^{cl}) + \sum_{i=1}^{N-1} \tilde l(x^*_{i|k},u^*_{i|k}) =- \tilde l(x_k^{cl}, u_k^{cl}) \nonumber\\
&+\! \sum_{i=1}^{i_0-1}\tilde l(x^*_{i|k},u^*_{i|k}) + \sum_{i=i_0}^{i_M-1}\tilde l(x^*_{i|k},u^*_{i|k}) + \sum_{i=i_M}^{N-1}\tilde l(x^*_{i|k},u^*_{i|k}). \label{eq:stadlim}
\end{align}
Furthermore, using Assumption \ref{as:cheap_block_turn} with \eqref{eq:prop_1a}-\eqref{eq:prop_1b}, there exists an admissible control sequence $\hat u_{0}$, $\hat u_{1}$, $\hat u_{2}, \ldots, \hat u_{{2M-1}}$ with corresponding feasible state trajectory  $\hat x_{0}$, $\hat x_{1}$, $\hat x_{2}, \ldots, \hat x_{{2M-1}}$, $\hat x_{{2M}}$ such that $\hat x_{0} = x^*_{i_0|k}\in\Bscr_{\delta}(\yline)$, $\hat x_{{2M}} = x^*_{i_M|k}\in\Bscr_{\delta}(\yline)$, $\dist(\hat x_{j}, \yline) \leq \rho(\delta)$ for all $ j\in[1:2M-1]$ and
\begin{equation}\label{eq:nakata}
\sum_{j=0}^{2M-1} \tilde l(\hat x_{j}, \hat u_{j}) \leq \tilde l_{\bar\delta} + \alpha_{\tilde l}(\rho(\delta))+\beta(M), 
\end{equation}
where $\tilde l_{\bar\delta} = \max_{x\in\Bscr_{\bar\delta}} \m\m \tilde l(x,u)$ and $\delta\in(0,\bar \delta]$. Consider a control sequence $u_{N+M-1,x^*_{1|k}}\in\uline^{N+M-1}(x^*_{1|k})$ constructed in the following way: 
\begin{equation} \label{eq:boromir}
u_{N\!+M\!-1,x^*_{1|k}}(j) \!=\! \begin{cases}
u^*_{j+1|k} \ \text{if}\  j\in [0:i_0-2]\m, \\
\hat u_{{j-i_0+1}} \ \text{if}\ j\in[i_0-1:i_0+2M-2]\m,  \\
u^*_{j+i_M-(i_0+2M-1)|k} \ \text{if}\ j\!\in\![i_0\!+\!2M\!-\!1\!:\!N\!+\!M\!-\!2]. 
\end{cases}
\end{equation}
Hence, the control sequence from state $x_{1|k}^*=x^{cl}_{k+1}$ is 
\begin{align*}
u_{N+M-1,x^*_{1|k}} \coloneqq (&u^*_{1|k}, u^*_{2|k}, \ldots, u^*_{i_0-1|k}, \hat u_{0}, \hat u_{1}, \ldots, \\
&\hat u_{{2M-2}}, \hat u_{{2M-1}}, u^*_{i_M|k}, u^*_{i_M+1|k}, \ldots, u^*_{N-1|k}).
\end{align*}
Using the control sequence of length $N+M-1$ as constructed in \eqref{eq:boromir}, the optimal value function $x^{cl}_{k+1}\coloneqq x^*_{1|k}$ satisfies
\begin{align*}
\tilde V_{N+M-1}(x^*_{1|k}) &\leq \tilde J_{N+M-1}(x^*_{1|k}, u_{N+M-1,x^*_{1|k}}) \\
&=\sum_{j=1}^{i_0-1}\tilde l(x^*_{j|k}, u^*_{j|k}) + \sum_{j=0}^{2M-1}\tilde l(\hat x_j, \hat u_j) + \sum_{j=i_M}^{N-1} l(x^*_{j|k}, u^*_{j|k}) .
\end{align*}
\begin{figure*}[t]
    \centering
    \begin{subfigure}{0.49\linewidth}
        \centering
        \includegraphics[width=0.75\textwidth]{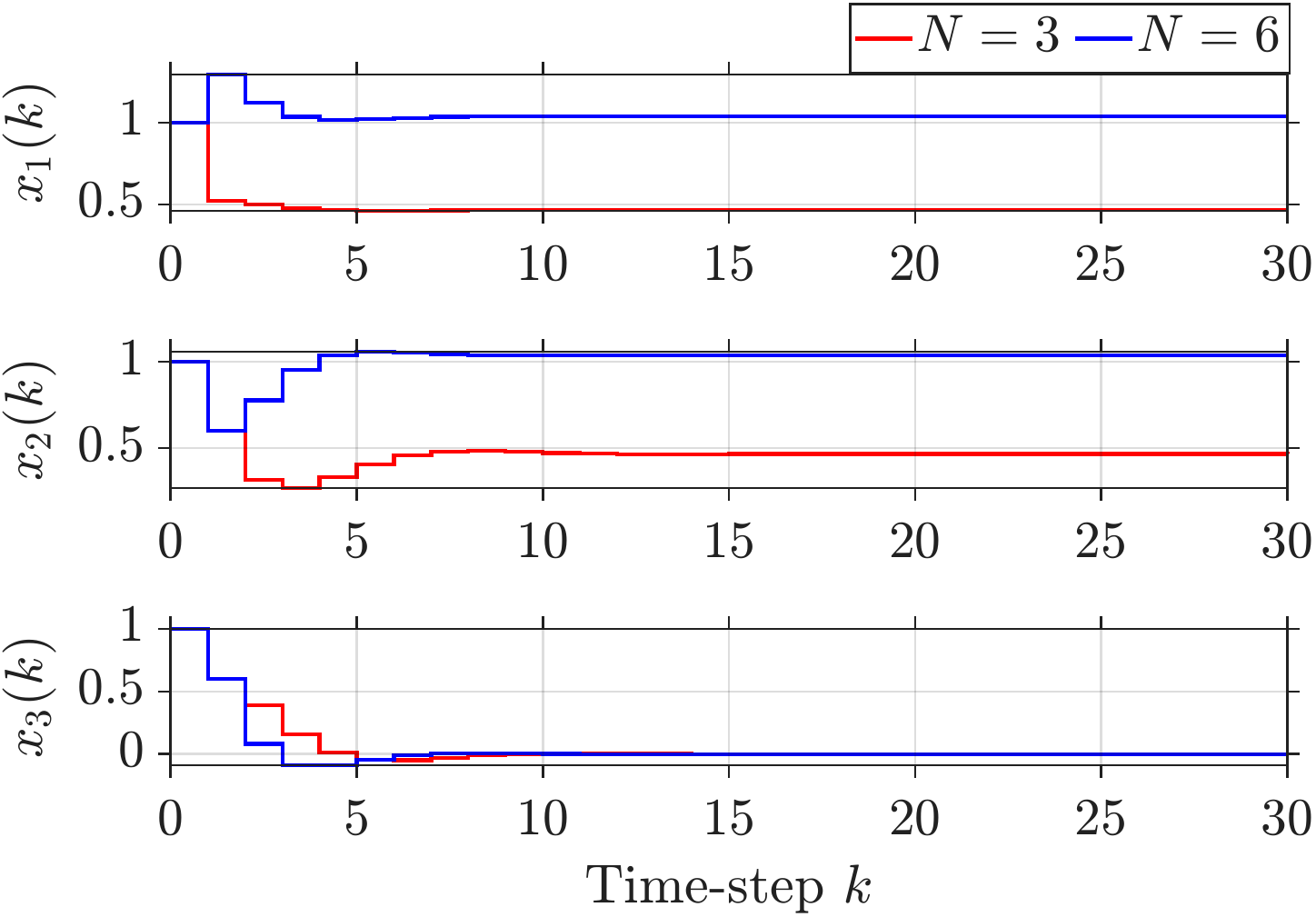}
        \caption{Closed-loop optimal state trajectories}
        \label{fig:CL_state_Lin_a}
    \end{subfigure}
    \hfill
    \begin{subfigure}{0.49\linewidth}
        \centering
        \includegraphics[width=0.75\textwidth]{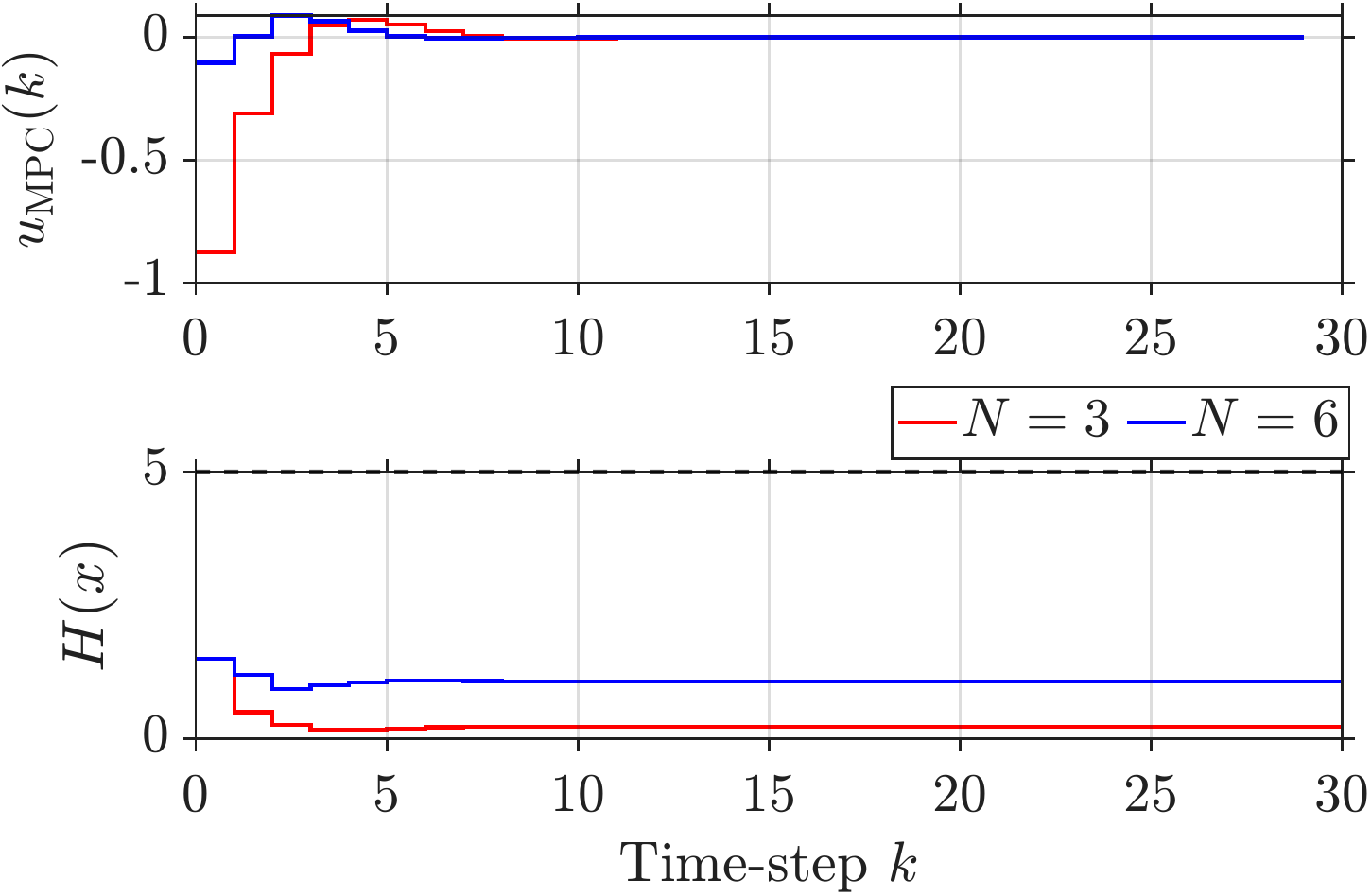}
        \caption{Closed-loop optimal input and $H(x)$}
        \label{fig:CL_state_Lin_b}
    \end{subfigure}

    \begin{subfigure}{0.49\linewidth}
        \centering
        \includegraphics[width=0.75\textwidth]{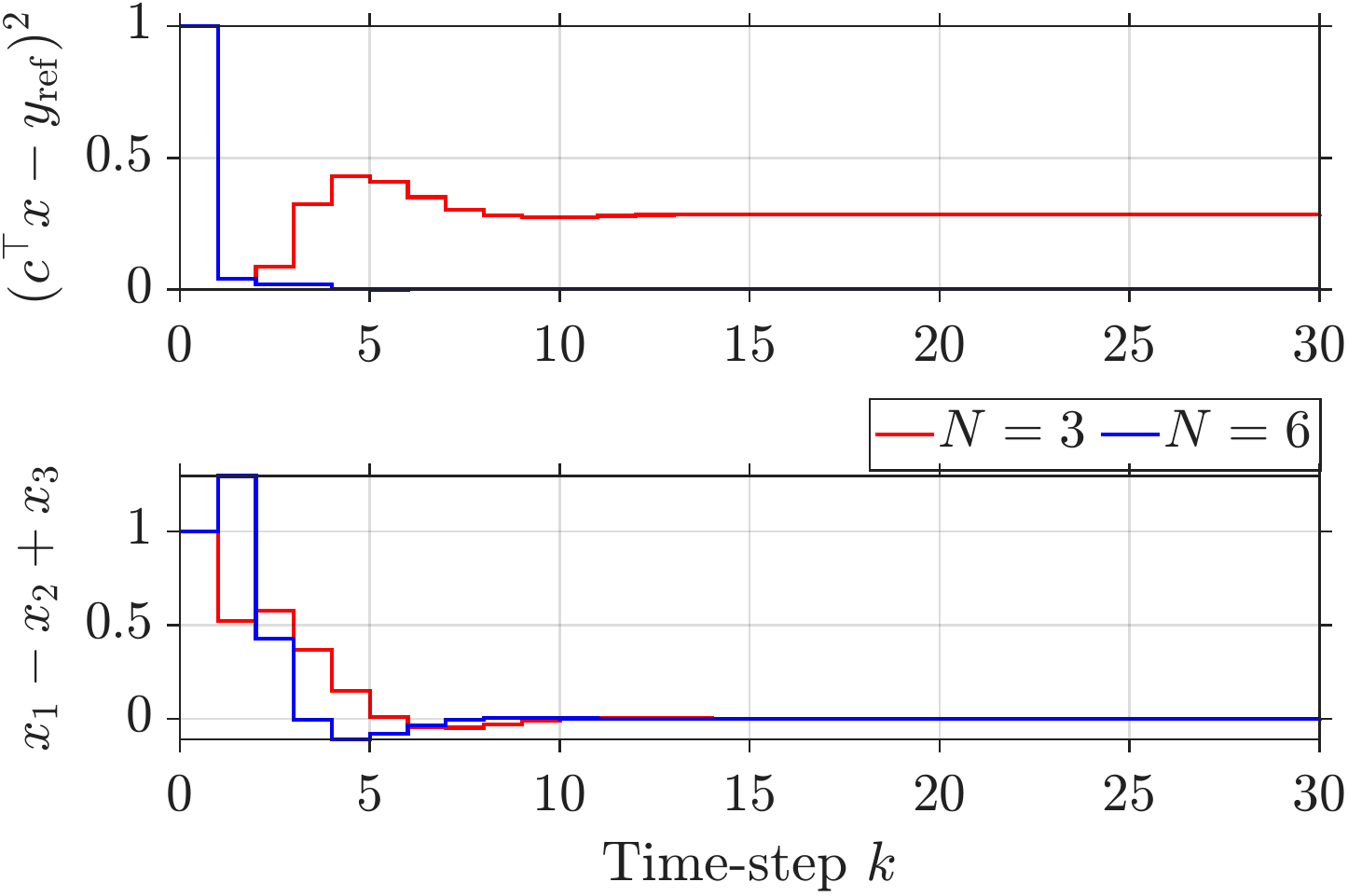}
        \caption{Individual components of $\JJJ_N$ in \eqref{eq:e_ocp}}
        \label{fig:CL_state_Lin_c}
    \end{subfigure}
    \hfill
    \begin{subfigure}{0.49\linewidth}
        \centering
        \includegraphics[width=0.75\textwidth]{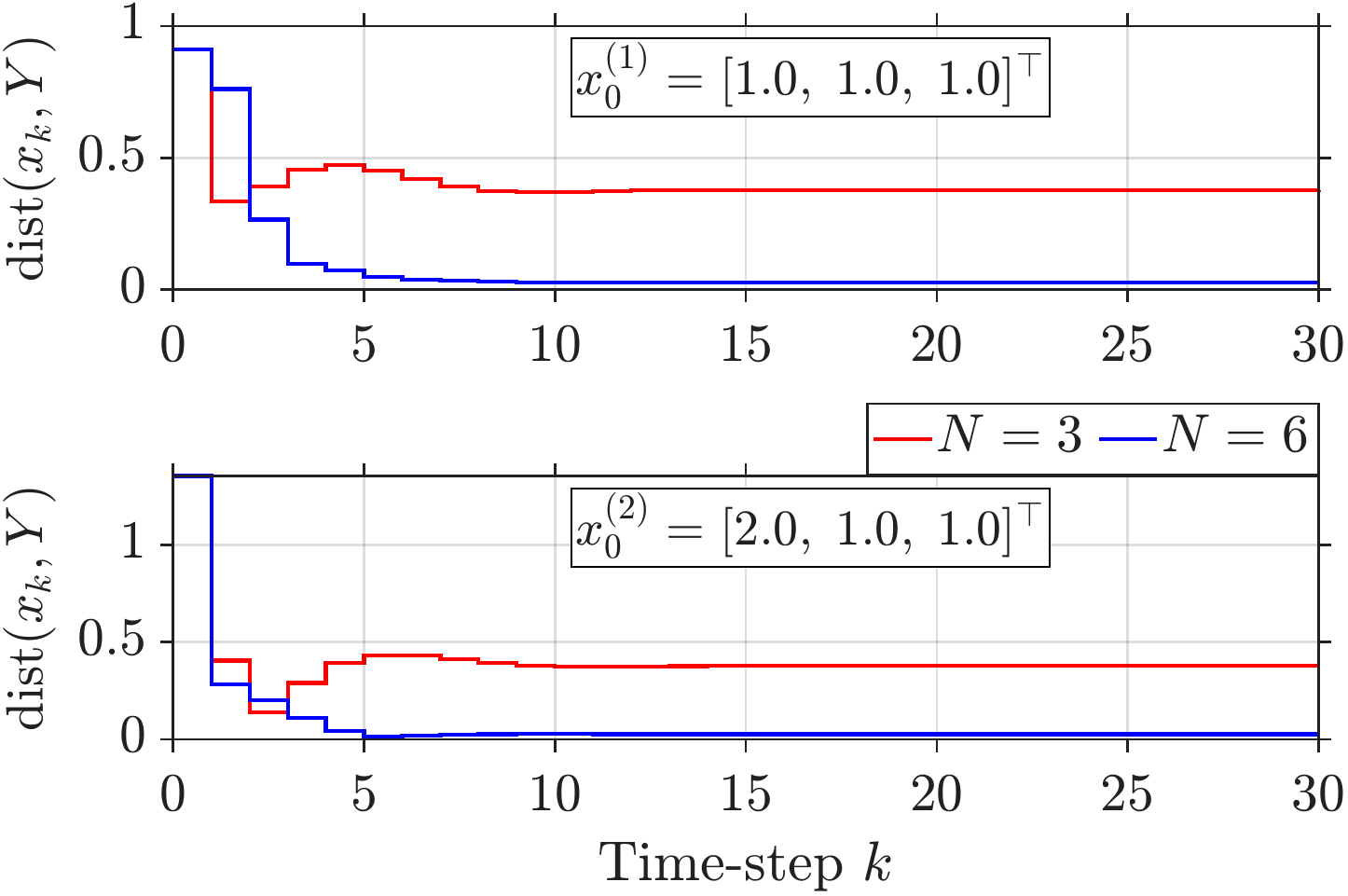}
        \caption{Turnpike for different initial conditions}
        \label{fig:CL_state_Lin_d}
    \end{subfigure}
    \caption{For Example 1, the plots are generated for two prediction horizons (i) $N=3$ (in red color) and (ii) $N=6$ (in blue color). The control constraint set is $[-5,\m 5]$, set $\xline_C$ is characterized by $C=5$, $c^\top = [0 \ 1 \ 1]$, and $y_{ref}=1$. For plots (a)-(c), the initial condition is chosen to be $x_0 = [1 \ 1\ 1]^\top$. Plot (d) shows the distance of the optimal state from the $\yline$ for two different initial conditions.}
    \label{fig:CL_state_Lin}
\end{figure*}

We rewrite the inequality above as
\begin{equation}\label{eq:legolas}
    \tilde V_{N}(x^{cl}_{k+1}) \leq \sum_{j=1}^{i_0-1}\tilde l(x^*_{j|k}, u^*_{j|k}) + \sum_{j=0}^{2M-1}\tilde l(\hat x_j, \hat u_j) + V_{N-i_{M}}(x^*_{i_M|k}) .
\end{equation}
using the following observations: $x^{cl}_{k+1} = x^*_{1|k}$ and
\begin{align}
 V_{N-i_M}(x^*_{i_M|k}) &= \sum_{j=i_M}^{N-1} l(x^*_{j|k}, u^*_{j|k})  \label{eq:baggins}\\
 V_{N}(x^*_{1|k})&\leq V_{N+M-1}(x^*_{1|k}) \label{eq:bilbo},
\end{align}
where \eqref{eq:baggins} is a result of principle of optimality which dictates that since $(x^*, u^*)$ is optimal over a horizon of length $N$ it must remain optimal for the horizon $N-i_M$ starting from $x^*_{i_M|k}$ and \eqref{eq:bilbo} is straightforward from  $\tilde l(x,u)\geq 0$ for all $x\in\xline_c$ and $u\in\uline$. On subtracting \eqref{eq:legolas} from \eqref{eq:stadlim}, we get 
\begin{align*}
   \tilde V_N(x^{cl}_{k+1}) - \tilde V_N&(x^{cl}_k)\leq  -\tilde l(x^{cl}_k, u_k^{cl}) +   \sum_{j=0}^{M-1}\Big[\tilde l(\hat x_j, \hat u_j)- \tilde l(x^*_{i_j|k}, u^*_{i_j|k})\Big]\nonumber\\[0.1em] 
   &+ \sum_{j=M}^{2M-1}\tilde l(\hat x_j, \hat u_j) + V_{N-i_{M}}(x^*_{i_M|k}) - \sum_{j=i_{M}}^{N-1}\tilde l(x^*_{j|k}, u^*_{j|k}).%
\end{align*}
Using \eqref{eq:baggins} and the fact that the rotated stage cost function is nonnegative, the aforementioned inequality simplifies to
\begin{equation*}\label{eq:tano}
    \tilde V_N(x^{cl}_{k+1}) - \tilde V_N(x^{cl}_k)\leq  -\tilde l(x^{cl}_k, u_k^{cl}) +   \sum_{j=0}^{M-1} \tilde l(\hat x_j, \hat u_j) + \sum_{j=M}^{2M-1}\tilde l(\hat x_j, \hat u_j).
\end{equation*}
Using \eqref{eq:nakata} in the above inequality, we get 
$$\tilde V_N(x^{cl}_{k+1}) - \tilde V_N(x^{cl}_k)\leq  -\tilde l(x^{cl}_k, u_k^{cl}) + \tilde l_{\bar\delta} + \alpha_{\tilde l}(\rho(\delta))+\beta(M),$$
which is \eqref{eq:gandalf} with $\epsilon(M,\delta, l_{\bar\delta})= \tilde l_{\bar\delta} + \alpha_{\tilde l}(\rho(\delta))+\beta(M).$  
\end{proof}
\begin{remark}[Choice of the set $\xline_C$]
In presence of additional state constraint $\xline$ such that $x_k\in\xline$ for each $k\in\nline_0$ should hold for \eqref{eq:e_ocp}, we can fix $\xline_c$ by selecting $C = \sup\{C_1\in(0,\infty)\mid X_{C_1}\subseteq\xline\}$. Hence, the constraint $\xline_C$ in \eqref{eq:e_ocp} becomes an approximation of the additional constraint set $\xline$. Such an approximation will yield conservative results and comes with a clear trade-off. Clearly, increasing $C$ will enlarge $\xline_C$ and improve the approximation $\xline$ but enlarging $X_C$ will lead to larger $M$ in Assumption \ref{as:reach_Y}. Hence an increase in choice of horizon $N$ in Theorem \ref{th:Gr13_Th4.2}. 
\end{remark}
\begin{figure*}[t]
    \centering
    \begin{subfigure}{0.49\linewidth}
        \centering
        \includegraphics[width=0.75\textwidth]{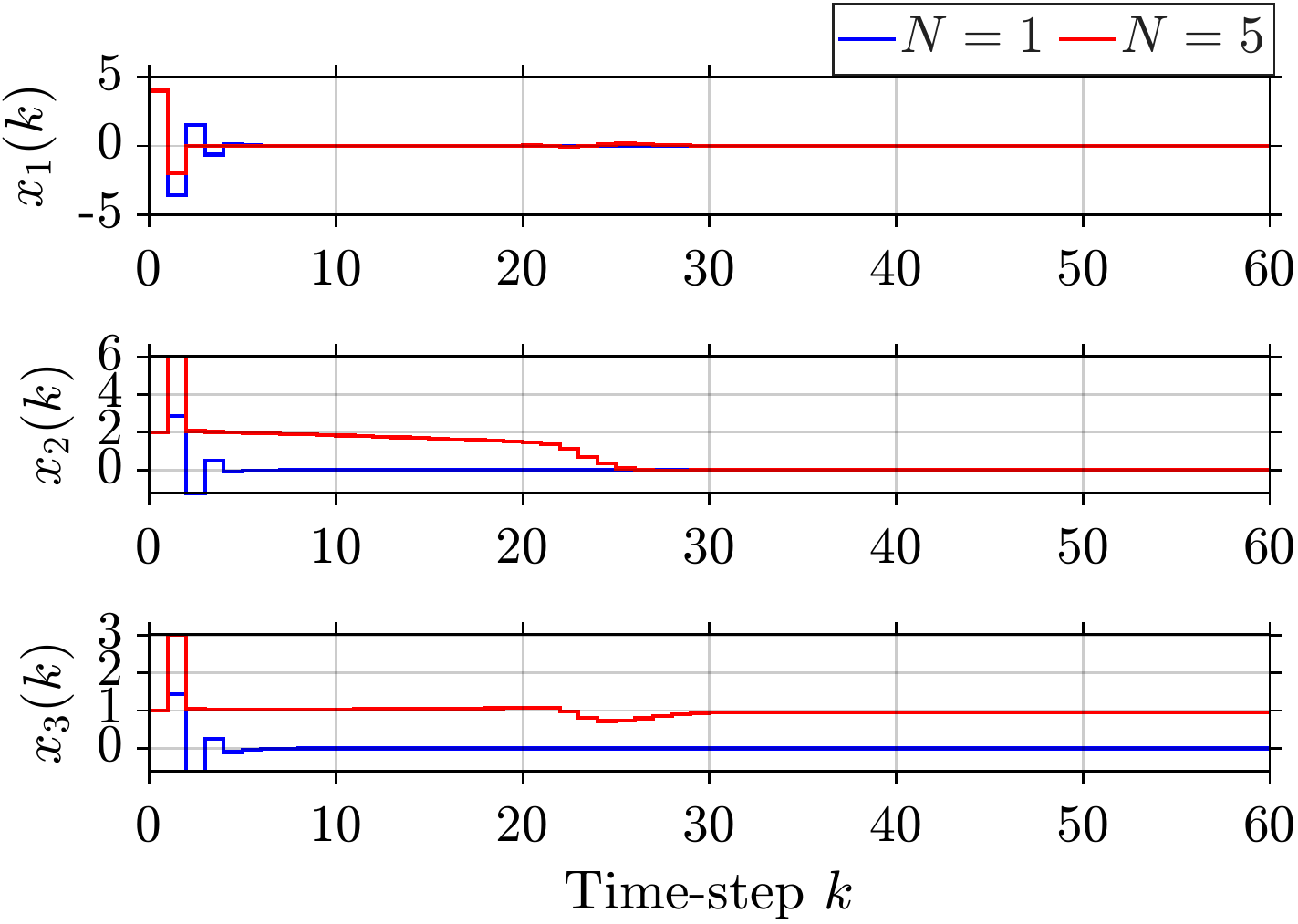}
        \caption{Closed-loop optimal state trajectories}
        \label{fig:CL_state_a}
    \end{subfigure}
    \hfill
    \begin{subfigure}{0.49\linewidth}
        \centering
        \includegraphics[width=0.75\textwidth]{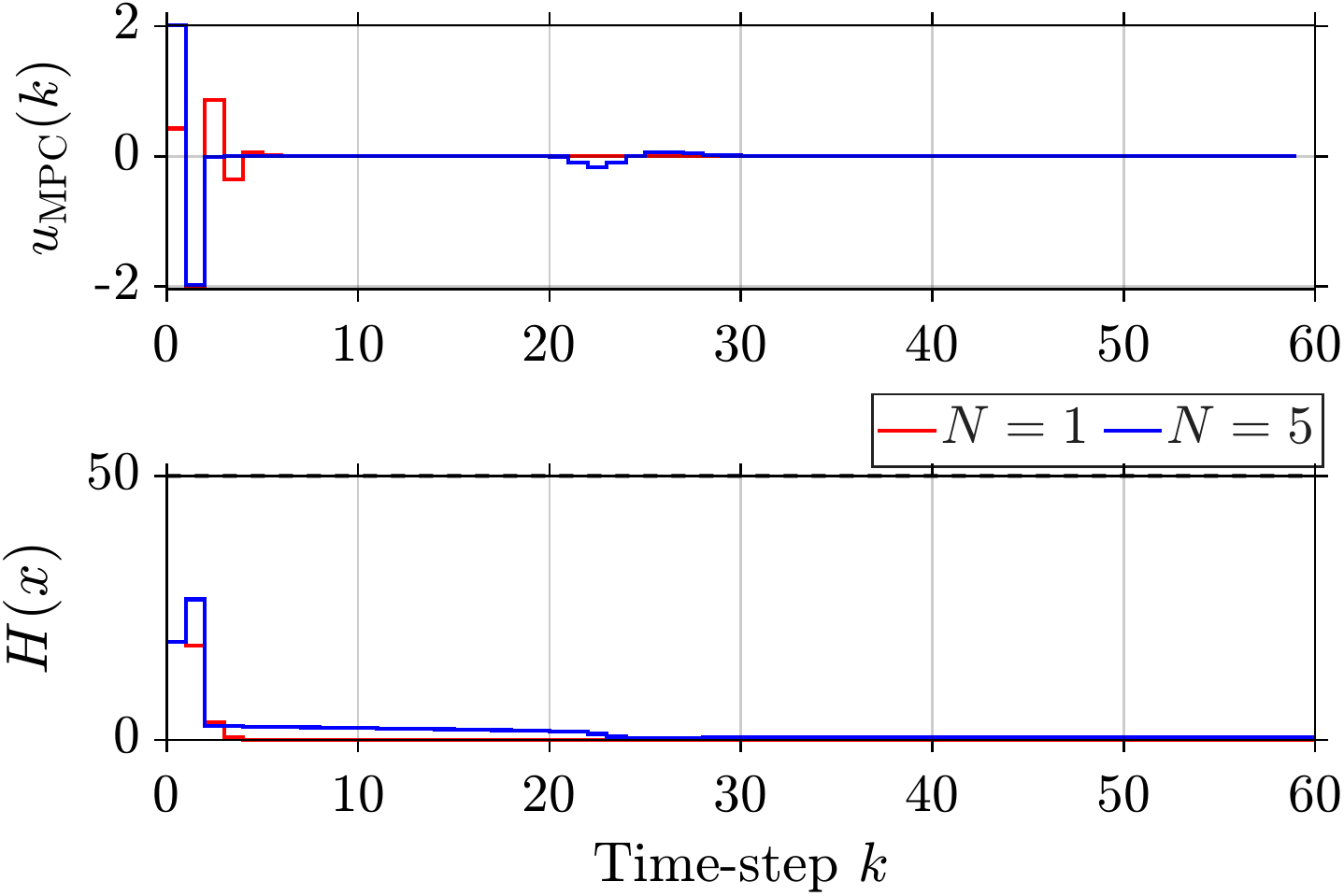}
        \caption{Closed-loop optimal input and Hamiltonian}
        \label{fig:CL_state_b}
    \end{subfigure}

    \begin{subfigure}{0.49\linewidth}
        \centering
        \includegraphics[width=0.75\textwidth]{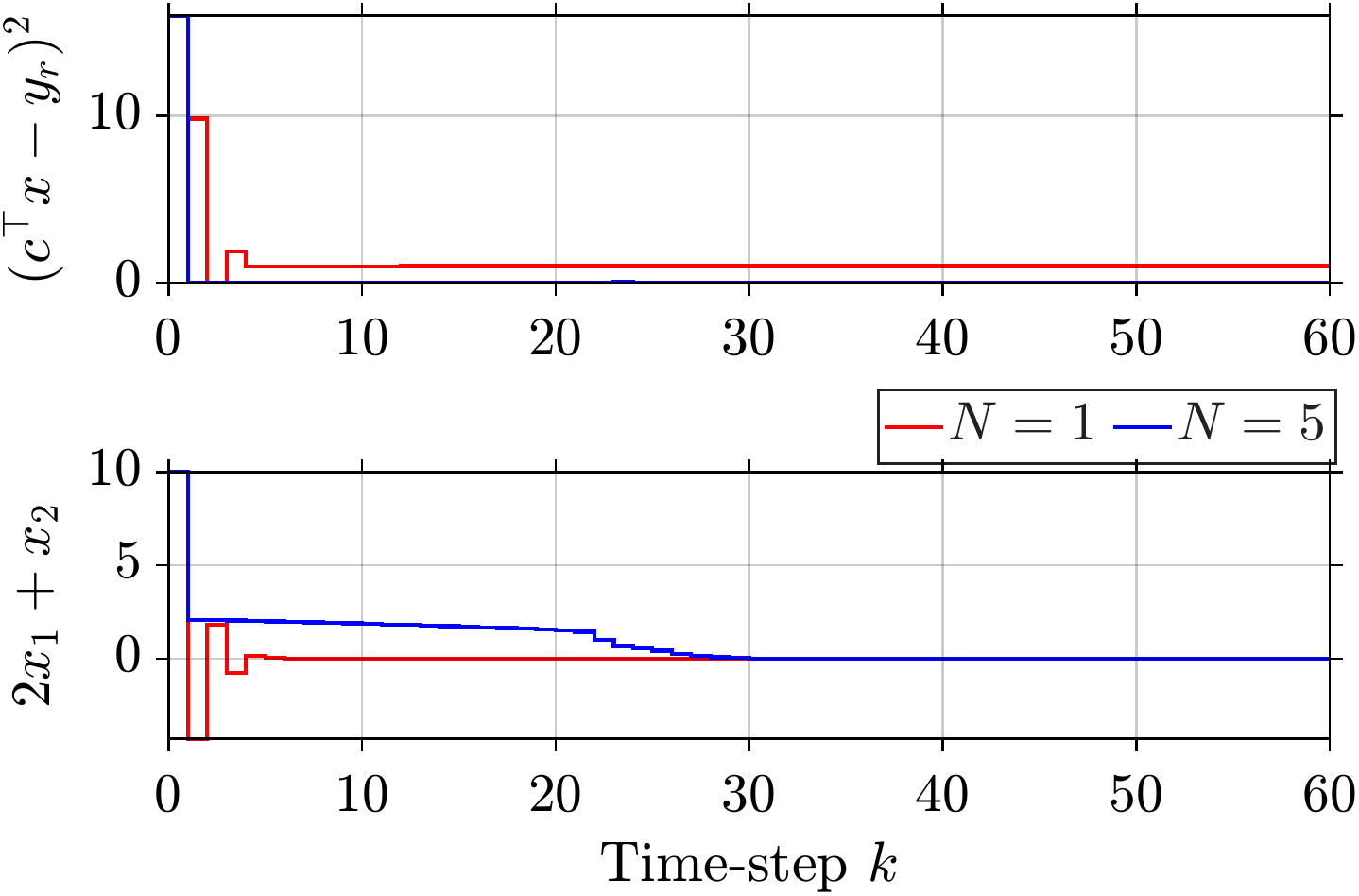}
        \caption{Individual components of cost $\JJJ_N$}
        \label{fig:CL_state_c}
    \end{subfigure}
    \hfill
    \begin{subfigure}{0.49\linewidth}
        \centering
        \includegraphics[width=0.75\textwidth]{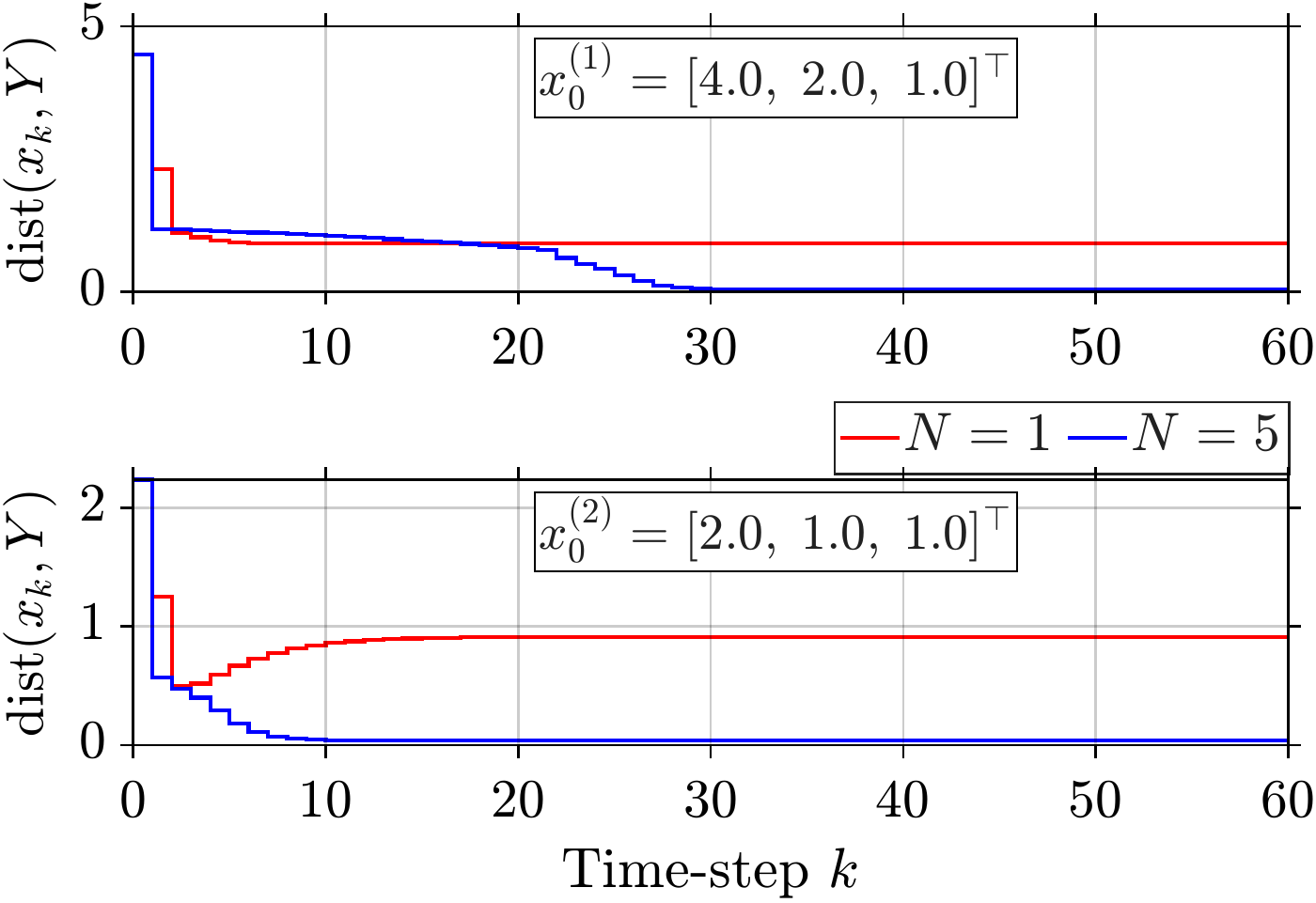}
        \caption{Turnpike for two different initial conditions}
        \label{fig:CL_state_d}
    \end{subfigure}

    \caption{For Example 2, plots~(a)-(d) are generated by running the algorithm for two prediction horizons (i) $N=1$ (in red) and  (ii) $N=5$ (in blue). We set $\uline =[-5,\m 5]$, $\xline_C$ is characterized by $C=50$,  $c^\top = [1 \ 0 \ 1]$, and $y_{ref}=1$. For plots (a)-(c), the initial condition is chosen to be $x_0 = [4 \ 2\ 1]^\top$. Plot (d) shows the distance of the optimal state from $\yline$ for two different initial conditions.}
    \label{fig:CL_state}
\end{figure*}

\section{Numerical examples}\label{se:Mot_num_Ex}
\ \ \ In this section, we consider two numerical example to illustrate the results developed in the previous section. In Example 1 we consider a linear discrete-time pH system and in Example 2 we simulate the results for a nonlinear pH system. To perform the numerical simulations, we use \emph{CasADi}, \cite{An:19}.
\subsection{Example 1: Linear pH system}
Consider a linear discrete-time pH system described by \eqref{eq:phs_ddr_state}-\eqref{eq:phs_ddr_output} with 
\begin{equation*}
    J_1(x) = \bbm{1.5 & -0.5 & -0.5 \\ -0.5  & 1.5 & 0.5 \\ 0.5 & -0.5 & 1.0}, \ 
    J_2(x) = \bbm{0.5 & 0.5 & 0.5 \\  0.5  &  0.5 &  -0.5 \\-0.5  &  0.5  &  1.0 },
\end{equation*}
$Q=I_3$, $b^\top = \bbm{1& 0 &0}$, $c^\top = \bbm{0& 1& 1}$, and $y_{ref}=1$. A straightforward computation shows that $\zline = \{z_1,z_2,z_3\in\rline\mid z_1-z_2+z_3 = 0\}$ and $\xline_r = \{z_1,z_2,z_3\in\rline\mid z_2+z_3 = 1\}$. We fix $u_{\max}= 5$, set $C=5$ and define the set $\xline_C$ accordingly. Note that \emph{Assumption \ref{as:feasiblity_track_ddr}} holds.  Using these sets, define $\yline$ as per \eqref{eq:yline}:
\begin{align}\label{eq:yline_ex_1}
\yline &=\Big\{\bbm{2\alpha-1& \alpha& 1-\alpha}^\top\in\rline^3 \mid \alpha\in[-0.76, 1.75]\Big\}.
\end{align}
It can be checked that every point in $\yline$ is returnable to $\yline$ within six-steps using an admissible control sequence and the total accumulated rotated cost of such an excursion has a uniform upper bound. Hence \emph{Assumption \ref{as:tube_Y}} holds. The simulation plots for this example are shown in Figure \ref{fig:CL_state_Lin} and Figure \ref{fig:3D_state_lin}. Note that for $N=3$, the tracking error remains positive, while it goes to zero with $N=6$, as seen in Figure \ref{fig:CL_state_Lin}(c). The closed-loop state and input trajectory satisfy the constraint of associated \eqref{eq:e_ocp} at all time-instants. The phase plot depicting the convergence of the closed-loop trajectories is depicted in Figure \ref{fig:3D_state_lin} for both the prediction horizons.  
 \begin{figure}[h]
    \centering
    \includegraphics[width=0.65\textwidth]{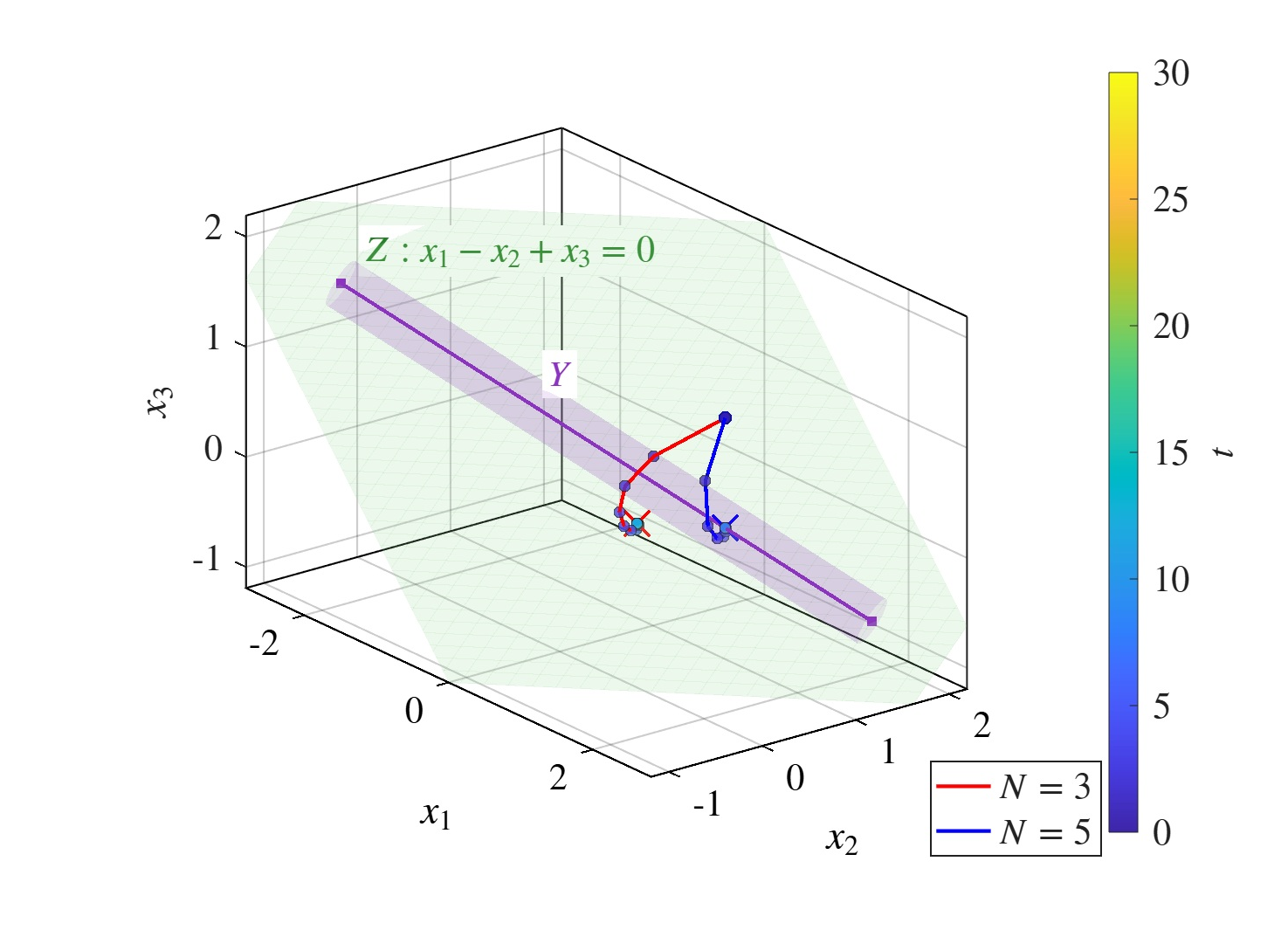} \vspace{-5mm}
    \caption{For Example 1 in Section \ref{se:Mot_num_Ex}, we show the phase plot of the closed-loop system \eqref{eq:cl_MPC} obtained by applying the receding horizon algorithm. The conservative manifold $\zline$ (depicted in green color) is described by equation $x_1-x_2+x_3=0$ and the set $\yline$ (depicted in purple color) is defined in \eqref{eq:yline_ex_1}.}
    \label{fig:3D_state_lin}
\end{figure}
\subsection{Example 2: Nonlinear pH system}
Consider a discrete-time nonlinear pH system described by \eqref{eq:phs_ddr_state}-\eqref{eq:phs_ddr_output} 
with
$$J_1(x) = \bbm{1 + A(x_{1,2}) & -\tfrac{1}{2} & 0\\ 1 & 1 & 0 \\ 0 & 0 & 1},\ J_2(x) = \bbm{1-A(x_{1,2}) & \tfrac{1}{2} & 0\\ -1 & 1 & 0 \\ 0 & 0 & 1},$$
where $x_{1,2}=[x_1, x_2]^\top$ and $A(x_{1,2}) = \tfrac{1}{4}(4\|x_{1,2}\|^2+1)^2$. Hence,
$$J_1^{-1}(x) = \frac{1}{1 + A(x_{1,2}) + \tfrac{1^2}{2}}\bbm{1 & \tfrac{1}{2} & 0\\ -1 & 1+A(x_{1,2}) & 0 \\ 0& 0 & 1}.$$
To identify the conservative manifold $\zline$, we compute 
$$R^{\tfrac{1}{2}}(x)QJ_1^{-1}(x)x =  \frac{4\|x_{1,2}\|^2+1}{1 + \tfrac{1}{4}(4\|x_{1,2}\|^2+1)^2 + \tfrac{1}{2}}\bbm{x_1+\tfrac{1}{2}x_2\\0\\0}.$$
We fix $\uline = [-5,5]$. Clearly, $\uline$ is a compact and convex set with $0\in\operatorname{int}\uline$. 
We obtain $\zline = \{(x_1,x_2, x_3)\in\rline^3\mid 2x_1+x_2 = 0\}$ which turns out to be a linear subspace. Next, we consider the finite-horizon optimal control problem stated in \eqref{eq:e_ocp} with $c = \bbm{1 & 0 & 1}^\top$ and $y_{ref} = 1$. Hence, $\xline_r = \{x\in\rline^3 \mid x_1+x_3 = 1\}$.  We fix $\xline_C = \{x\in\rline^3\mid 2x_1^2+x_2^2+x_3^2 \leq 100\}$ and use \eqref{eq:yline} to compute \vspace{-3mm}
\begin{equation}
\yline = \Big\{\bbm{\alpha&-2\alpha&1-\alpha}^\top \mid \alpha\in[-3.60, 3.91]\Big\} \label{eq:yline_ex_2} \vspace{-2mm}
\end{equation}
Clearly, $\yline$ is nonempty as $x = [0, 0, 1]^\top\in\yline$ and compact. The map $R^{\nicefrac{1}{2}}(x)QJ_1^{-1}(x)x$ is twice continuously differentiable. Hence \emph{Assumption \ref{as:feasiblity_track_ddr}} holds. For any $z\in\yline$, we have $J_1^{-1}(z)J_2(z)z = \bbm{-\alpha & -2\alpha& 1-\alpha}^\top$ with $\alpha\in[-3.60, 3.91]$. Letting $u = \alpha$, we get that 
$J_1^{-1}(z)J_2(z)z+bu = \bbm{0&0&1}^\top\in\yline$.
Hence, for this example the set $\yline$ is control invariant. Infact, with $u=\alpha$ any $z\in\yline$ is driven to $\bbm{0&0&1}^\top$ in a single time step. Hence \emph{Assumption \ref{as:tube_Y}} holds with $M_1(\delta)=1$ and the accumulated rotated cost being zero and we can set $\beta_0(L)=0$. Consequently, the bound in Theorem \ref{th:meas_turn_Y} becomes uniform with respect to horizon length $N$ and we retrieve the standard measure turnpike property for this example. 
 \begin{figure}[h]
    \centering
    \includegraphics[width=0.65\textwidth]{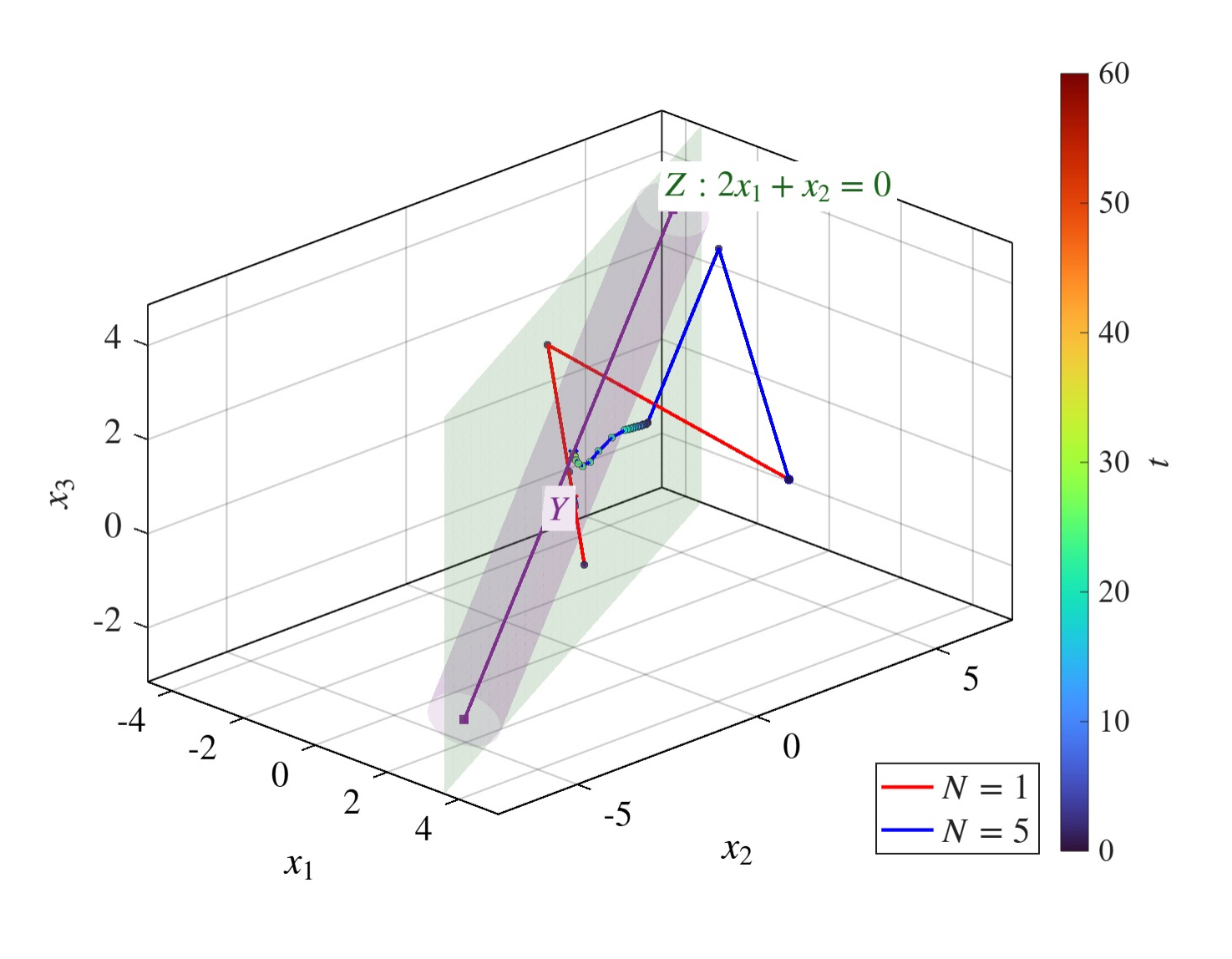} \vspace{-4mm}
    \caption{For Example 2, we show the phase plot of the closed-loop system \eqref{eq:cl_MPC} obtained by applying the receding horizon algorithm. The conservative manifold $\zline$ (depicted in green color) is described by equation $2x_1+x_2=0$ and the set $\yline$ (in purple color) is defined in \eqref{eq:yline_ex_2}.}
    \label{fig:3D_state}
\end{figure}
The simulation plots for this example are in Figure \ref{fig:CL_state} and Figure \ref{fig:3D_state}. Note that for $N=1$, the tracking error remains positive, while it goes to zero with $N=5$, as seen in Figure \ref{fig:CL_state}(c). The closed-loop state and input trajectory satisfy the constraint of associated \eqref{eq:e_ocp}, see Figure \ref{fig:CL_state}(b). Figure \ref{fig:CL_state}(d) shows phase-plot of the closed-loop trajectories. 
\section{Conclusion}\label{se:concl}
In this work, we address an energy-optimal output stabilization problem for discrete-time port-Hamiltonian system using economic model predictive control algorithm. The dynamics of the pH system are described in difference and differential form. Under the assumption that the optimal control problem is strictly dissipative with respect to a set $\yline$ and exploiting the underlying pH structure, we establish a weaker measure turnpike property of the open-loop optimal solutions. The presented turnpike property is required since the set $\yline$ is not guaranteed to be control invariant. Using the measure turnpike property, we establish closed-loop practical stability for closed-loop system obtained by applying a receding horizon algorithm. We illustrate our results using two numerical examples. {For future works, we will establish performance bounds on the closed-loop solutions of \eqref{eq:e_ocp} obtained using dissipativity-based MPC algorithm and extend the ideas to pH systems with state-dependent input map and non-quadratic Hamiltonian.}
\section*{Acknowledgment}\vspace{-3mm}
Author V.K. Singh acknowledges Dr. Arijit Sarkar for discussions during the initial stage of this research.
\bibliographystyle{elsarticle-num}
\bibliography{ref_SC}
\end{document}